\documentclass[a4paper,fleqn]{cas-dc}

\usepackage[numbers]{natbib}
\usepackage{xfrac}
\usepackage{units}
\usepackage{amsthm}
\usepackage{subcaption}
\usepackage{bbding}

\def\tsc#1{\csdef{#1}{\textsc{\lowercase{#1}}\xspace}}
\tsc{WGM}
\tsc{QE}
\tsc{EP}
\tsc{PMS}
\tsc{BEC}
\tsc{DE}

\newtheorem{definition}{Definition}
\newtheorem{proposition}{Proposition}
\newtheorem{theorem}{Theorem}

\begin{document}
\let\WriteBookmarks\relax
\def\floatpagepagefraction{1}
\def\textpagefraction{.001}
\let\printorcid\relax  
\shorttitle{MCFE-SI-NAS in Federated Learning}
\shortauthors{Ruyuan Zhang et~al.}

\title [mode = title]{Multi-client Functional Encryption for Set Intersection with Non-monotonic Access Structures in Federated Learning}                      



\author[1]{Ruyuan Zhang}[type=editor]
\ead{ruyuanzhang@seu.edu.cn}

\affiliation[1]{organization={School of Cyber Science and Engineering, Southeast University},
                city={Nanjing},
                postcode={210096}, 
                country={China}}

\author[1]{Jinguang Han}[style=Chinese]
\cormark[1]
\ead{jghan@seu.edu.cn}
\cortext[cor1]{Corresponding author}

\begin{abstract}
	Federated learning (FL) based on cloud servers is a distributed machine learning framework that involves an aggregator and multiple clients, which allows multiple clients to collaborate in training a shared model without exchanging data.
	Considering the confidentiality of training data, several schemes employing functional encryption (FE) have been presented.
	However, existing schemes cannot express complex access control policies. 
	In this paper, to realize more flexible and fine-grained access control, we propose a multi-client functional encryption scheme for set intersection with non-monotonic access structures (MCFE-SI-NAS), where multiple clients co-exist and encrypt independently without interaction.
	All ciphertexts are associated with an label, which can resist "mix-and-match" attacks.
	Aggregator can aggregate ciphertexts, but cannot know anything about the plaintexts.
	We first formalize the definition and security model for the MCFE-SI-NAS scheme and build a concrete construction based on asymmetric prime-order pairings.
	The security of our scheme is formally proven.
	Finally, we implement our MCFE-SI-NAS scheme and provide its efficiency analysis.
\end{abstract}




\maketitle

\section{Introduction}

	Federated Learning (FL) \cite{Kairouz2021} is a promising paradigm that has attracted extensive attention due to its advantages that a shared model is  trained collaboratively by a aggregator with multiple private inputs while ensuring raw data secure.
	However, FL paradigm still suffers from serious security issues, particularly inference attacks and sensitive data leakage problems.
	To tackle the above problems, several secure FL frameworks have been presented:
	FL based on secure multi-party computation (SMPC) \cite{Chen2024} \cite{Liu2024}, FL based on homomorphic encryption (HE) \cite{Gong2024} \cite{Yan2024} \cite{Zhang2020} and FL based on functional encryption (FE) \cite{Chang2023} \cite{Feng2024} \cite{Qian2024}. 
	FL based on SMPC frameworks employ secret sharing technology to share training parameters among multiple parties. 
	An aggregator interacts with a certain number of parties for decrypting and model training, which increases communication cost and disconnection risk. 
	FL based on HE frameworks support arithmetic operations on ciphertexts that allows update global model parameters with encrypted local gradients, and then aggregator obtains an encryption of the result.
	However, aggregator need extra interactions to recover the result.
	FE is a novel encryption technology equipped with the same features of computation over ciphertexts as SMPC and HE, but has an advantage over other solutions that no interaction is required.
	In a FE scheme, a decryption key is associated with a function, where authorized users directly decrypt ciphertexts and obtain the function values of encrypted data without disclosing any other information about encrypted data.
	Multi-client functional encryption (MCFE), where multiple clients encrypt data separately, have been applied in FL  to protect data confidentiality and train models \cite{Chang2023} \cite{Qian2024} \cite{Nguyen2023}.
	In an FL based on MCFE framework, multiple clients generate independently ciphertexts with their private inputs and an authorized aggregator collects clients' ciphertexts to aggregate data and train models.
	 
 	Model training usually requires the aggregation of datasets with same sample space and different feature domains.
 	However, in reality, it is nearly impossible to find two raw datasets from different clients that share same space.
	Hence, sample alignment is an important data preparation operation in model training. 
	Private set intersection (PSI) technology \cite{Liu2023PSI} \cite{Liu2020PSI} \cite{Yang2019PSI} \cite{He2022} \cite{Angelou2020} has been used to solve the above problem, which allows two or more parties exchange encrypted massage with each other to compute intersection of their private sets without revealing anything else, but it has a disadvantage that additional interactions between parties are required.
	Inspired by PSI protocols, MCFE for set interaction (MCFE-SI) scheme was proposed \cite{Kamp2019Two}, where a third party is responsible for calculating the set intersection of two clients’ sets without interacting with clients.  
	
	In addition, model poisoning attack \cite{Bagdasaryan2020} is an attack mode over FL global models, where malicious aggregators can directly influence global parameters and perform backdoor tasks. 	
	According to \cite{Li2023}, model poisoning attack seriously threatens availability of FL, because any unauthorized aggregator may substitute global models with malicious models for strengthening the poisoning effect.
	Hence, to prevent unauthorized aggregators from participating in model training, access control to training data is significant.
	To realize access control on encrypted data, MCFE with access control schemes \cite{Nguyen2022} \cite{Agrawal2021MultiParty} have been proposed, but only support monotonic access structures, which can not meet more complex access requirements.
	Especially, in the complex FL environment, more expressive non-monotonic access structures are desirable and must be supported, but unfortunately has not been considered in existing MCFE-SI schemes.
	
	In this paper, we propose a MCFE-SI with non-monotonic access structures scheme, where aggregators' decryption keys embed fined-grained access policies and ciphertexts of each client are associated with an attribute set.
	Intersections can be calculated correctly if and only if the attribute set matches access policies of aggregators.
	To meet complex data access requirements in the FL, the proposed scheme supports more expressive non-monotonic access structures that can express any policy.

\subsection{Related Work}
\subsubsection{Functional Encryption}

	Waters \cite{Sahai2008} first introduced the concept of of FE which addresses the "all-or-nothing" issue (i.e., a decryptor is either able to recover the entire plaintext, or nothing) in public-key encryption schemes.
	Concretely, there exists a trusted authority TA responsible for generating a key $sk_{f}$ for a specified function $f$. 
	When given a ciphertext $CT_{x}$ and $sk_{f}$, the key holder learns the functional value $f(x)$ and nothing else.
	O'Neill \cite{ONeill2010} and Boneh et al. \cite{Boneh2011} provided formal definitions and security models for FE. 
	In the multi-user cases, Goldwasser et al. \cite{Goldwasser2014} first provided the definition of the multi-input functional encryption (MIFE), which supports multiple parties independently encrypt their data.
	However, MIFE schemes are vulnerable to "mix-and-match" attack since any client's ciphertext can be combined for decryption computation.
	For instance, suppose that two clients respectively encrypt $\{x_0,x_1\}$ and $\{y_0,y_1\}$, and a evaluator can calculate $f(x_{\mu_0},y_{\mu_1})$ for any combination of $\mu_0,\mu_1 = \{0,1\}$ , which leads to too much leakage.
	To resist this attack, multi-client functional encryption \cite{Shi2023} was proposed, where a label is applied to encrypt messages. As a result, ciphertexts can be combined to decrypt if and only if they contain the same label.
	
	There exists an inherent issue in MCFE schemes that a secret key can be used to recover the functional values of all ciphertexts.
	To address this problem, Abdalla et al. \cite{Abdalla2020} first proposed a FE scheme with fine-grained access control that combines attribute-based encryption (ABE) with FE for inner product (FEIP). 
	Inspired by the scheme \cite{Abdalla2020}, Nguyen et al. \cite{Nguyen2022} presented a duplicate-and-compress technique to transform a single-client FE scheme with access control into corresponding MCFE schemes.
	Dowerah et al. \cite{dowerah2024sacfe} designed an attribute-based functional encryption scheme which realizes fine-grained access control structures through monotone span programs, and supports to encrypt messages with unbounded length. 
	
	The above schemes require a fully trusted authority to generate keys. 
	Datta et al. \cite{Datta2023Decentralized}  proposed a decentralized multi-authority attribute-based inner-product FE scheme to remove the trusted authority. 
	Similarly, Agrawal et al. \cite{Agrawal2021MultiParty} presented an multi-authority FE scheme with linear secret-sharing structures based on composite-order bilinear maps. Unfortunately,  computation cost of composite-order bilinear maps is expensive. 
	The above FE with access control schemes realized monotonic access structures that contain "AND" gate, "OR" gate and threshold strategy, but did not address non-monotonic access structures.
	
\subsubsection{MCFE for Set Intersection}

	MCFE schemes for set intersection (MCFE-SI) was proposed first by Kamp et al. \cite{Kamp2019Two}. 
	However, set intersection in the scheme \cite{Kamp2019Two} can be publicly recovered by anyone.	
	To solve this issue, Lee et al. \cite{Lee2022} designed a concrete MCFE-SI scheme in asymmetric bilinear groups which is proved static security under their introduced assumptions.
	In \cite{Lee2022}, there exist $n$ clients and an evaluator, where each client encrypts their set  with an label and outsources the encrypted set to the evaluator.
	The evaluator receiving a functional key can calculate the set intersection from chiphertexts.
	Lee \cite{Lee2023} later proposed three efficient MCFE-SI schemes via a ciphertext indexing technology.
	Rafee \cite{Rafiee2023} presented a flexible MCFE-SI scheme, where discrete logarithm calculations are required for computing the final set intersection. 	
	However, the above MCFE-SI schemes do not consider access control problems.
	
\subsubsection{FE for Federated Learning}

	Qian et al. \cite{Qian2022} proposed a cloud-based privacy-preserving federated learning (PPFL) aggregation scheme based on FE, which is efficient in aggregation phase.
	In order to remove a trusted third party, Qian et al.\cite{Qian2024} later proposed a decentralized MCFE scheme for FL, which supports non-interactive partial decryption keys generation and client dropout.
	Chang et al. \cite{Chang2023} applied a dual-mode decentralized MCFE to design a new framework of PPFL, which prevents the private information of target users from being recovered by aggregator through uploading local models. 
	Feng et al. \cite{Feng2024} present a multi-input functional proxy re-encryption scheme for PPFL, which allows a semi-trusted central server to aggregate parameters without obtaining the intermediate parameters and aggregation results.
	
	The main differences between our scheme and the schemes \cite{Chang2023} \cite{Feng2024} \cite{Qian2024}  \cite{Qian2022} are as follows: 
	(1) our scheme can resist "mix-and-match" attacks, but the scheme \cite{Feng2024} is unable to address it; 
	(2) our scheme focus on set intersection operation, while the schemes \cite{Chang2023} \cite{Feng2024}  \cite{Qian2024}  \cite{Qian2022} execute inner product operation;  
	(3) our scheme can support fine-grained access control, while access issue is not considered in the schemes \cite{Chang2023} \cite{Feng2024}  \cite{Qian2024}  \cite{Qian2022}.
	
	We compare the properties of our MCFE-SI-NAS scheme with related schemes in Table \ref{tbl1}, in terms of function, access structures, resistance to "mix-and-match" attack and bilinear group. N/A denotes not applicable.   
	
	\begin{table*}[cols=5,pos=h]		
		{\fontfamily{ptm}\selectfont
		\caption{The comparision between our scheme and related schemes}
		\centering
		\label{tbl1}
		\begin{tabular}{|c|c|c|c|c|}
			\hline 
 			Schemes & Function & Access structures & Resistance to “mix-and-match” attack & Bilinear group \\
			\hline 
			\cite{Abdalla2020} & Inner product & Monotonic & N/A & Prime-order \\
			\hline
			\cite{Nguyen2022} & Inner product & Monotonic & \checkmark & Prime-order  \\
			\hline
			\cite{Datta2023Decentralized} & Inner product & Monotonic & N/A & Prime-order \\
			\hline
			\cite{dowerah2024sacfe} & Inner product & Monotonic & N/A & Prime-order \\
			\hline
			\cite{Agrawal2021MultiParty} & Inner product & Monotonic & N/A & Composite-order \\			
			\hline 
			\cite{Kamp2019Two} & Set intersection & \XSolid & \checkmark & Prime-order \\			
			\hline
			\cite{Lee2022} & Set intersection & \XSolid & \checkmark & Prime-order \\			
			\hline
			\cite{Lee2023} & Set intersection & \XSolid & \checkmark & Prime-order \\			
			\hline
			\cite{Rafiee2023} & Set intersection & \XSolid & \checkmark & Prime-order \\			
			\hline
			Our scheme & Set intersection & Non-monotonic & \checkmark & Prime-order \\			
			\hline
		\end{tabular}
		
		}
	\end{table*}

\subsection{Our Contributions}

	Non-monotonic access structure is important in real application.
	For instance, the documents of history department might be encrypted with the attributes: "Year:2024", "Department:history". 
	An aggregator who is authorized to aggregate data of historical departments but prohibited to access data of biological departments, and hence his/her decryption keys are related with the policy: "Year:2024" AND "Department:history" NOT "Department:biology".
	However, monotonic structures cannot express the above policy. 
	Non-monotone access structures is more expressive.
	In terms of the above problems, we first propose an MCFE-SI with non-monotonic access structures (MCFE-SI-NAS) scheme which can realize any policy including "AND", "OR", "NOT" as well as threshold policy.
	Our scheme enables each client to encrypt independently and upload data in a non-interaction manner.
	
	The contributions of our MCFE-SI-NAS scheme are as follows.
	
	(1) The proposed scheme allows multiple clients co-exist and encrypt their data independently, and all ciphertexts are bound with a label for resisting "mix-and-match" attack.
	
	(2) Our scheme also supports non-monotonic access structures that can realize any access structures over attributes.
	
	(3) Ciphertext indexing technology can be used to find intersections of ciphertext without decrypting. Aggregator can aggregate ciphertexts and output the set intersection of any two client plaintexts, but cannot learn anything about plaintexts. 
	
	(4) We first provide the definition and security model of our MCFE-SI-NAS scheme, and build concrete construction on asymmetric bilinear groups. The security proof of the MCFE-SI-NAS scheme is formally given. We implement and evaluate our MCFE-SI-NAS scheme, and provide efficiency analysis.

%
	
\subsection{Organization}
	
	The rest of this paper is organized as follows. 
	Section 2 shows the preliminaries used in this paper.
	In Section 3, we present the concrete construction of our MCFE-SI-NAS scheme.
	The security proof and implementation are described in Section 4 and Section 5, respectively.
	Section 6 concludes this paper.

\section{Preliminaries}
	The preliminaries used in this paper are introduced in this section. 
	Table \ref{tbl2} shows all symbols applied in this paper.	
	\begin{table}[width=.9\linewidth,cols=2,pos=h]
		
		\caption{Syntax}
		\label{tbl2}
		{\fontfamily{ptm}\selectfont 
 		\begin{tabular*}{\tblwidth}{@{\extracolsep{\fill}} p{2cm} p{5cm} }
			\toprule
			\text{Notions} & \text{Explanations} \\
			\midrule			
			$1^{\lambda}$ & A security parameter \\
			$d$ & The size of attribute in ciphertext \\
			$sk$ & A secret key \\
			$msk$ & Master secret keys \\
			$csk_k$ & $k$-th client's encryption keys \\
			$pp$ & Public parameters \\
			$\mathbb{A}$ & Monotonic access structures \\
			$\widetilde{\mathbb{A}}$ & Non-monotonic access structures \\			
			$N$ & The number of clients \\
			$f$ & An index function \\					
			$\mathcal{A}$ & A PPT adversary \\
			$\mathcal{C}$ & A challenger \\
			$\mathcal{B}$ & A simulator \\	
			$S$ & An attribute set \\				
			$\Upsilon$ & A set intersection \\
			$SK_{\widetilde{\mathbb{A}},f}$ & Decryption keys \\			
			$Tag$ & A label \\
			$M_k$ & A message set held by $k$-th client \\
			$CT_k$ & Ciphertexts corresponding to $M_k$ \\	 
			$\mathcal{HS}$ & Honest client sets \\
			$\mathcal{CS}$ & Corrupted client sets \\
			PPT & Probabilistically polynomial time \\		    
			FE & Functional encryption \\
			SI & Set intersection \\
			FL & Federated learning \\
			MCFE & Multi-client functional encryption \\		    
			MCFE-SI-NAS & MCFE for SI with non--monotonic  access structures \\			
			\bottomrule
		\end{tabular*}
	}
	\end{table}

\subsection{Bilinear Groups}
	
	\begin{definition} 
		\normalfont
		$G,\hat{G}$ and $G_T$ denote three cyclic groups with prime order $p$.
		$e:G \times \hat{G} \rightarrow G_T$ is a bilinear map if it satisfies the following properties \cite{BonehIBE}.
		
		(1) Bilinearity. If $g \in G$ and $\hat{g} \in \hat{G}$, the equation
		$e(g^x,\hat{g}^y) = e(g^y,\hat{g}^x) = e(g,\hat{g})^{xy}$ holds for any $x,y \in Z_p$.
		
		(2) Non-generation. For any $g \in G$ and $\hat{g} \in \hat{G}$, $e(g,\hat{g}) \neq 1$. 
		
		(3) Computability. $e(g,\hat{g})$ can be computed efficiently for any $g \in G$ and $\hat{g} \in \hat{G}$.
		
		$\mathcal{BG}(1^\lambda)\rightarrow(G,\hat{G},G_T,e,p)$ denotes a generator of bilinear groups, which inputs a security parameter $1^{\lambda}$ and outputs bilinear groups $(G,\hat{G},G_T,e,p)$.
		There are three types of pairings: Type-I, Type-II and Type-III. Type-III pairing provides good performance and is efficient.
		We select the Type-III pairing to build our MCFE-SI-NAS scheme in this paper to improve its efficiency.
	\end{definition}

\subsection{Complexity Assumptions}

	We utilize the assumptions introduced by Lee \cite{Lee2023} to prove the security of the proposed scheme.
	The complexity assumptions are defined as dynamic assumptions depending on the key queries of the adversary.
	
	We first define a function $J(N,\nu^*,\mathcal{Q})$ for demonstrating subsequent security proof.	
	Set $N$ be a positive integer and $\nu^* \in [N]$ be a targeted index.
	$\mathcal{Q}=\{(w,v)\}$ denotes a set of index pairs such that $w,v\in[N]$ and $w<v$.
	Suppose an index set 
	$J=\{
	\nu : 1 \leq \nu \neq \nu^* \leq N \vert  (\nu,\nu^*)\notin \mathcal{Q}  \text{ if } \nu<\nu^*$ and $(\eta^*,\nu)\notin \mathcal{Q} 
	\text{ if } \nu^* < \nu \}$.
	For generating a set $J$, the function $J(N,\nu^*,\mathcal{Q})$ is defined as follows .	
	\begin{flushleft}
		\centering
		\begin{tabular}{l}		
			\hline	
			\textbf{Function} $J(N,\nu^*,\mathcal{Q}) \text { where } \mathcal{Q}=\{(w, v)\}$ \\			
			Set $J=\emptyset$. \\
			For each $\nu \in\{1, \ldots, N\} \backslash\{\nu^*\} \text { : }$\\
			\quad $\text { If } \nu<\nu^* \text { and }(\nu, \nu^*) \notin \mathcal{Q} \text {, then add } \nu \text { to } J \text {. }$ \\
			\quad Add $\text { If } \nu>\nu^* \text { and }(\nu^*, \nu) \notin \mathcal{Q} \text {, then add } \nu \text { to } J \text {. }$	 \\
			Output $J \text {. }$ \\
			\hline
		\end{tabular}
	\end{flushleft}	
		
	For example, suppose $N=5$, $\nu^*=2$ and $\mathcal{Q}=\{(1,5),(2,4),(3,4),(2,5)\}$, 	it can obtain $J=\{1,3\}$ since $(1,2)\notin \mathcal{Q}, (2,3)\notin \mathcal{Q}, (2,4)\in \mathcal{Q} \text{ and } (2,5)\in \mathcal{Q}$.
	
	\begin{definition}
		\normalfont
		Let $\mathcal{BG}(1^{\lambda})\rightarrow(G,\hat{G},G_T,e,p)$, $N,\nu^*,\mathcal{Q},J$ be defined above. $g,\hat{g}$ denote generators of $G,\hat{G}$, respectively.
		Given the following tuple  	
		\begin{equation*}
			D=\left( 
			\begin{array}{c}
				g, \hat{g}, g^a,\left\{g^{b_w}\right\}_{w=1}^N,\left\{g^{a b_\nu}\right\}_{\nu \in J}, \\
				\left\{\left(\hat{g}^{b_w c_{w, v}}, \hat{g}^{b_v c_{w, v}}, \hat{g}^{1 /\left(b_w+b_v\right)}\right)\right\}_{(w, v) \in \mathcal{Q}} 
			\end{array}
			\right) \text{ and } Z,
		\end{equation*}				
		we say that the assumption holds on $(g,\hat{g},G,\hat{G},G_T,e,p)$ if all PPT adversary $\mathcal{A}$ can distinguish $Z=Z_0=g^{ab_{\nu^*}}$ and random $Z=Z_1 \in G$ with the following negligible advantage $\epsilon(\lambda)$:
		\begin{equation*}
			\left|\operatorname{Pr}\left[\mathcal{A}\left(D,Z_0\right)=1\right]-\operatorname{Pr}\left[\mathcal{A}\left(D, Z_1\right)\right]=1\right| \leq \epsilon(\lambda)
		\end{equation*}
		
	\end{definition}
	
	\begin{definition}
		\normalfont
		($q$-Decision Bilinear Diffie-Hellman Exponent Assumption in Symmetric Parings \cite{SONG2021})  Let $\mathcal{BG}(1^{\lambda})\rightarrow(G,\hat{G},G_T,e,p).$  
		Set $G=\hat{G}$ and $g,h$ be generators of $G$. Given the following tuple  	 
		\begin{equation*}
			D_1=(h,g,g^{\gamma},g^{(\gamma^{2})},\ldots,g^{(\gamma^{q})},g^{(\gamma^{q+2})},\ldots,g^{(\gamma^{2q})}) \text{ and } H,
		\end{equation*}
		we say that the assumption holds on symmetric group $(g,G,G_T,e,p)$ if all PPT adversary $\mathcal{A}$ can distinguish $H=H_0=e(g,g)^{\gamma^{q+1}}$ and $H=H_1\in G_T$ with the following negligible advantage $\epsilon(\lambda)$:
		$$
		\left|\operatorname{Pr}\left[\mathcal{A}\left(D_1,H_0\right)=1\right]-\operatorname{Pr}\left[\mathcal{A}\left(D_1, H_1\right)\right]=1\right| \leq \epsilon(\lambda).	$$	
		
	\end{definition}
	
	\begin{definition}
		\normalfont
		(The Variant of the $q$-Decision Bilinear Diffie-Hellman Exponent Assumption in Asymmetric Parings ($q$-DBDHE))
		Let $\mathcal{BG}(1^{\lambda})\rightarrow(G,\hat{G},G_T,e,p).$ $g,h$ denote generators of $G$ and $\hat{g},\hat{h}$ are generators of $\hat{G}$.
		Given the following tuple 
		\begin{equation*}
			D_2= \left( 
			\begin{array}{c}
				(h,g,g^{\gamma},g^{(\gamma^2)},...,g^{(\gamma^q)},g^{(\gamma^{q+2})},...,g^{(\gamma^{2q})}, \\
				\hat{g},\hat{h},\hat{g}^{\gamma},\hat{g}^{(\gamma^2)},...,\hat{g}^{(\gamma^q)},\hat{g}^{(\gamma^{q+2})},...  ,\hat{g}^{(\gamma^{2q})})
			\end{array}
			\right) \text{ and } T, 	
		\end{equation*}
		we say that the variant of $q$-DBDHE assumption holds on asymmetric group $(h,g,\hat{g},G,\hat{G},G_T,e,p)$ if all PPT adversary $\mathcal{A}$ can distinguish $T=T_0=e(g,\hat{h})^{\gamma^{q+1}}$ and $T=T_1 \in G_T$ with the following negligible advantage $\epsilon(\lambda)$:
		$$
		\left|\operatorname{Pr}\left[\mathcal{A}\left(D_2,T_0\right)=1\right]-\operatorname{Pr}\left[\mathcal{A}\left(D_2, T_1\right)\right]=1\right| \leq \epsilon(\lambda).	$$	
	\end{definition}

\subsection{Access Structures}

	Let a set of parties $P=\{P_1,...,P_n\}$.
	A collection $\mathbb{A}$ is said to be monotone if $B \in \mathbb{A} \text{ and } B \subseteq C$, then $C \in \mathbb{A}$.
	A monotonic access structure is a monotonic collection $\mathbb{A} \subseteq 2^{P} \backslash \{\emptyset\}$.
	The sets in $\mathbb{A}$ are called the authorized sets and those not in $\mathbb{A}$ are unauthorized sets.

\subsection{Linear Secret-Sharing Schemes}

	Let $\mathcal{L}$ be a share-generating matrix for $\prod$. $\mathcal{L}$ is equipped with $o$ rows and $c$ columns.
	$P$ denotes a set of parties.
	Let $\rho:\{1,\ldots, o\} \rightarrow P$ be a mapping that maps a row of $\mathcal{L}$ to a party.
	A secret-sharing scheme $\prod$ over a set of parties $P$ is called linear secret-sharing scheme (over $Z_p$) if it contains the following algorithms:
	\begin{itemize}
		\item Share: it inputs a secret $\alpha \in Z_p$ and selects randomly $s_2,...,s_c \in Z_p$. Let $\overrightarrow{\zeta}=(\alpha,s_2,...,s_c)^\top$. It outputs  $\mathcal{L} \overrightarrow{\zeta}$ as the vector of $o$ shares of the secret $\alpha$. The share $\lambda_i = (\mathcal{L}\overrightarrow{\zeta})_i$ belongs to a party $\rho(i)$.
		
		\item Recon: it inputs a set $S \in \mathbb{A}$, and sets $I = \{i | \rho(i)\in S\}$.
		There exists a set of constants $\{\pi_i\}_{i\in I}$ satisfying that $\sum_{i\in I} \pi_{i} \cdot \lambda_i=\alpha$.	
	\end{itemize}
	
	\begin{proposition}
		\label{proposition}
		\cite{Goyal2006} A vector $\overrightarrow{vec}$ is linearly independent of a series of vectors represented by a matrix $\mathcal{L}$ if and only if there is a vector $\overrightarrow{v}$ satisfying that $\mathcal{L} \cdot \overrightarrow{v} = \overrightarrow{0}$ and $\overrightarrow{vec}\cdot \overrightarrow{v}=1$.
	\end{proposition}
	
\subsection{System Model}

	The framework of our MCFE-SI-NAS scheme is shown in the Figure \ref{FIG:1}. 
	Our system model contains four types of entities, namely a trusted authority TA, a aggregator, $N$ clients $\{CL_1, CL_2, \ldots, CL_N\}$ and a cloud server CSP.
	TA is a fully trusted party that is responsible for issuing encryption keys to clients, and calculates secret keys for aggregator with specified access policy and function.
	Each client $CL_k$ works independently and encrypts their own data sets using the encryption keys from TA.
	All clients' ciphertexts are uploaded to CSP in a non-interactive manner.
	When receiving decryption keys from TA, the aggregator executes computation over ciphertexts and model training.

	\begin{figure*}
		\centering
		\includegraphics[width=1.6\columnwidth]{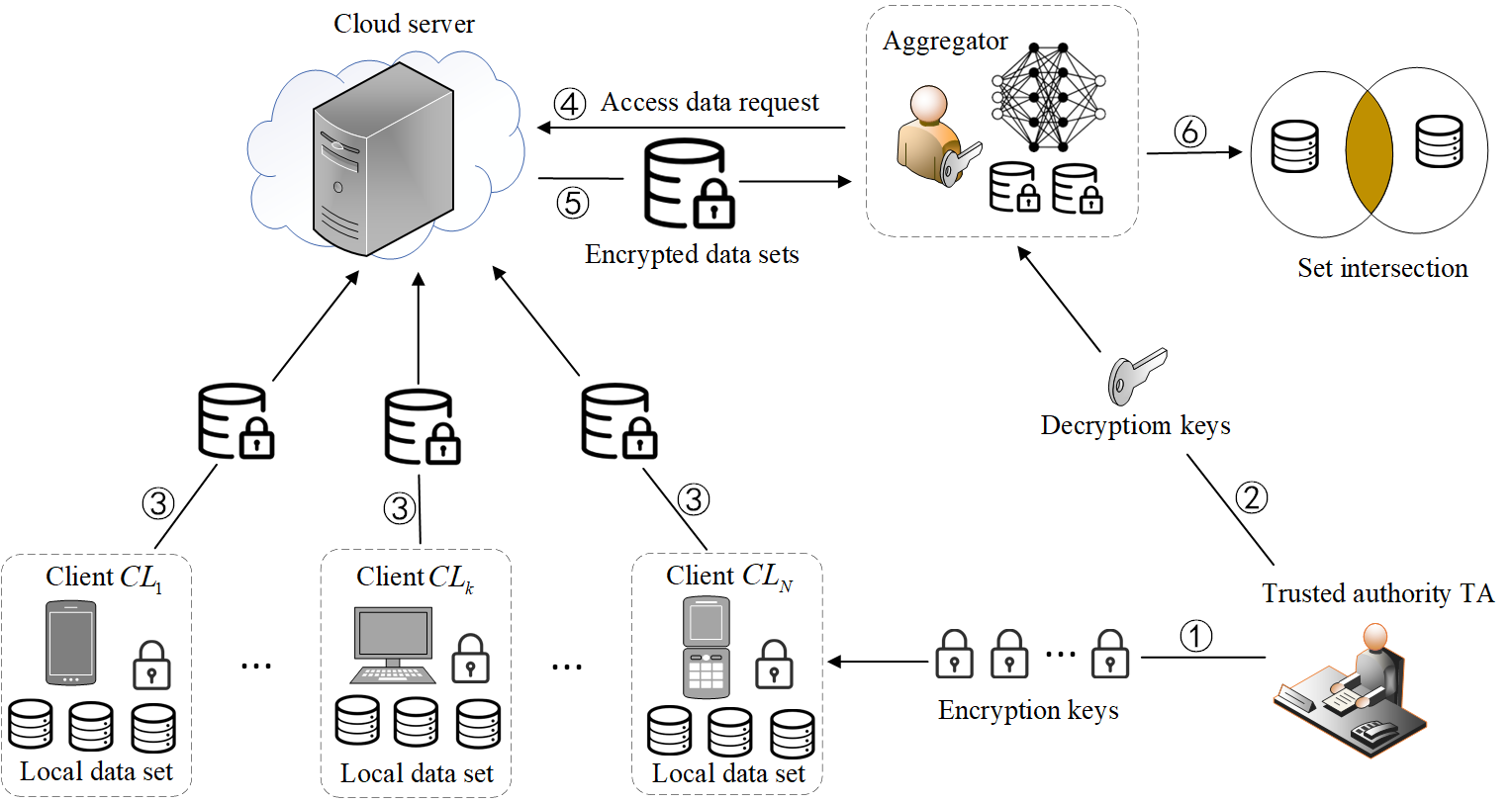}
		\caption{System model of our MCFE-SI-NAS scheme.}
		\label{FIG:1}
	\end{figure*}	
	
	Our MCFE-SI-NAS scheme for label space $\mathcal{T} \in G_T$ contains the following four algorithms.
	
	$Setup(1^\lambda,d,N) \rightarrow (sk,csk_1,csk_2,...,csk_N,pp)$.
	The algorithm is executed by the TA, and takes as input a security parameter $1^{\lambda}$, a preset attribute number $d$ and a client number $N$. 
	It outputs a secret key $sk$, client encryption keys $csk_1,csk_2,...,csk_N$ and public parameters $pp$, 
	where master secret keys are $msk = \{sk,csk_1,csk_2,...,csk_N\}$.
	
	$KeyGen(msk,pp,f,\widetilde{\mathbb{A}}) \rightarrow SK_{\widetilde{\mathbb{A}},f}$.
	The algorithm is executed by the TA.
	It takes as input master secret keys $msk$, public parameters $pp$, an index function $f$ and an access structure $\widetilde{\mathbb{A}}$. Then, it outputs decryption keys $SK_{\widetilde{\mathbb{A}},f}$ corresponding to $\widetilde{\mathbb{A}}$ and $f$.
	
	$Enc(pp,S,Tag,M_k,csk_k) \rightarrow CT_k$.
	The algorithm is executed by each client $CL_k$, where $k \in [N]$.
	It takes as input public parameters $pp$, a set of attributes $S$, a label $Tag \in \mathcal{T}$, a massage set $M_k$ and corresponding client encryption keys $csk_k$. It outputs ciphertexts $CT_k$.
	
	$Dec(pp,CT_w,CT_v,SK_{\widetilde{\mathbb{A}},f}) \rightarrow M_{w}\bigcap M_{v}/\perp$.
	The algorithm is executed by aggregator.
	It takes as input public parameters $pp$, clients' ciphertexts $CT_k$, decryption keys $SK_{\widetilde{\mathbb{A}},f}$ and
	outputs an intersection set $M_{w}\bigcap M_{v}$ or a special
	symbol $\perp$ denoting a failure.

	\begin{definition}
		\normalfont
		(Correctness)
		A multi-client functional encryption for set intersection with 	non-monotonic access structures (MCFE-SI-NAS) scheme is correct if
		$$
		Pr \left[
		\begin{array}{c|l}
			&Setup(1^\lambda,d,N) \rightarrow (sk, \\
			Dec(pp,CT_w, &csk_1,csk_2,...,csk_N,pp);\\			 
			CT_v, SK_{\widetilde{\mathbb{A}},f}) &KeyGen(msk,pp,f,\widetilde{\mathbb{A}}) \\ \rightarrow M_{w}\bigcap M_{v} &\rightarrow SK_{\widetilde{\mathbb{A}},f}; \\        
			& Enc(pp,S,Tag,M_k,csk_k) \\
			& \rightarrow CT_k
		\end{array}
		\right] = 1
		$$
	\end{definition}
	
\subsection{Security Model}
	
	We define the indistinguishability (IND) security \cite{Lee2023} \cite{Rafiee2023} \cite{Chotard2018} of our MCFE-SI-NAS scheme by using the following game executed between a adversary $\mathcal{A}$ and a challenger $\mathcal{C}$.
	
	\textbf{Init.} $\mathcal{A}$ initially submits an honest client set $\mathcal{HS} \subset [N]$ 
	and a set of corrupted clients $\mathcal{CS}=\{1,...,N\} \setminus \mathcal{HS}$.
	In addition, $\mathcal{A}$ selects a targeted label $Tag^*$, an attribute set $S^*$, 
	an index set $\mathcal{Q}=\{(w,v)\}$ of function key queries,
	two challenging massage sets $\{M_{1,0}^*,...,M_{N,0}^*\}$ and $\{M_{1,1}^*,...,M_{N,1}^*\}$ with the following restrictions. 
	
	(1) $w,v\in [N]$ for each $(w,v)\in\mathcal{Q}$. 
	
	(2) $CSI(\{M_{k,0}^*\}_{k\in [N]},\mathcal{Q}) = CSI(\{M_{k,1}^*\}_{k\in[N]},\mathcal{Q})$.
	
	\textbf{Setup.} $\mathcal{C}$ executes the algorithm $Setup(1^\lambda,d,N) \rightarrow (sk,csk_1,csk_2,...,csk_N,pp)$
	and generates the secret key $sk$, encryption keys $csk_1,...,csk_N$ and public parameters $pp$. 
	$\mathcal{C}$ keeps $(sk,\{csk_k\}_{k\in \mathcal{HS}})$ and sends $(pp,\{csk_k\}_{k\in \mathcal{CS}})$ to $\mathcal{A}$.
	
	\textbf{Phase-1.} $\mathcal{A}$ makes decryption key queries for function $f=(w,v)\in \mathcal{Q}$ and many access structures $\widetilde{\mathbb{A}}_q$, 
	where $S^*\notin \widetilde{\mathbb{A}}_q$ for all $q$. 
	$\mathcal{C}$ runs algorithm $KeyGen(msk,pp,\mathcal{Q},\widetilde{\mathbb{A}}_q)  \rightarrow SK_{\widetilde{\mathbb{A}}_q,\mathcal{Q}}$.
	Then $SK_{\widetilde{\mathbb{A}}_q,\mathcal{Q}}$ are sent to $\mathcal{A}$.
	
	\textbf{Challenge.} $\mathcal{C}$ flips a coin and obtains a bit $\mu\in\{0,1\}$.
	$\mathcal{C}$ executes the algorithm $Enc(pp,S^*,Tag^*,  M_{k,\mu}^*,csk_{k}) \rightarrow CT_{k,\mu}$ for each $k \in [N]$.
	The challenged ciphertexts $\{CT_{k,\mu}\}_{k\in[N]}$ are sent to $\mathcal{A}$.
	
	\textbf{Phase-2.} $\mathcal{A}$ continues to issue decryption key queries as in \textbf{Phase-1}.
	
	\textbf{Guess.} The guess $\mu'$ of $\mu$ is outputted by $\mathcal{A}$.  $\mathcal{A}$ wins the game if $\mu'=\mu$.
	
	We consider two weaker security notions \cite{Rafiee2023} for our MCFE-SI-NAS scheme:
	\begin{itemize}
		\item Passive security (P-IND). There is no corruption among clients, i.e., $\mathcal{HS}=[N]$ and $\mathcal{CS}=\emptyset$.
		\item Static security (S-IND). The corrupted client sets are chosen before init phase.
	\end{itemize}
	
	In this paper, we construct a MCFE-SI-NAS scheme with P-IND security level.
	 
	\begin{definition}	
		A MCFE-SI-NAS scheme is P-IND secure if and only if all PPT adversaries $\mathcal{A}$ win the above game with the following negligible advantage $\epsilon(\lambda)$:
		$$
		Adv_{\mathcal{A}} = |Pr[\mu=\mu'] - \frac{1}{2}| < \epsilon(\lambda).
		$$
	\end{definition}

\section{The Construction of Our MCFE-SI-NAS Scheme}

	In this section, we present the detailed construction of our MCFE-SI-NAS scheme which contains four algorithm, namely
	$Setup$, $KeyGen$, $Enc$ and $Dec$ algorithms which shown in Figure 2 $\thicksim$ Figure 5, respectively.
	
	\textit{Correctness.}
	Our MCFE-SI-NAS scheme is correct since the following equations hold.	
	\begin{equation*}
		\begin{split} 
			&\widetilde{sk}_{i,1}^{(1)} = sk_{i,1}^{(1)} \cdot \prod_{j=2}^{d} k_{i,j}^{{(1)}^{y_j}} \\
			&= g^{\lambda_i} \cdot u_0^{rt_i} \cdot (u_1^{-\frac{\theta_{i,2}}{\theta_{i,1}}} \cdot u_2)^{y_2 t_i} \cdot \cdot \cdot  (u_1^{-\frac{\theta_{i,d}}{\theta_{i,1}}} \cdot u_d)^{y_d t_i} \\
			&= g^{\lambda_i} \cdot u_0^{rt_i} \cdot (u_1^{-\frac{\theta_{i,1}y_1+\theta_{i,2}y_2+...+\theta_{i,d}y_d-\theta_{i,1}y_1}{\theta_{i,1}}}\cdot u_2^{y_2} \cdot \cdot \cdot u_d^{y_d})^{t_i} \\
			&= g^{\lambda_i} \cdot u_0^{rt_i} \cdot (u_1^{\frac{-<\overrightarrow{\theta},\overrightarrow{Y}>}{\theta_{i,1}}} \cdot u_1^{y_1} \cdot \cdot \cdot u_d^{y_d})^{t_i} \\
			&= g^{\lambda_i} \cdot (u_0^r \cdot u_1^{y_1} \cdot \cdot \cdot u_d^{y_d})^{t_i};
		\end{split}
	\end{equation*}
		
	\begin{equation*}
		\begin{split}
			&\frac{e(\widetilde{sk}_{i,1}^{(1)},ct_{1,w})}{e(ct_{2,w},sk_{i,2}^{(1)})} \\
			&= \frac{e(g^{\lambda_i} \cdot (u_0^r \cdot u_1^{y_1} \cdot \cdot \cdot u_d^{y_d})^{t_i},\hat{g}^{s_w})}{e((u_0^r \cdot \prod_{i=1}^{d}u_i^{y_i})^{s_w},\hat{g}^{t_i})} \\
			&= \frac{e(g^{\lambda_i},\hat{g}^{s_w}) \cdot e((u_0^r \cdot u_1^{y_1} \cdot \cdot \cdot u_d^{y_d})^{t_i},\hat{g}^{s_w})}{e((u_0^r \cdot u_1^{y_1} \cdot \cdot \cdot u_d^{y_d})^{s_w},\hat{g}^{t_i})} \\
			&= e(g,\hat{g})^{\lambda_i s_w};
		\end{split}
	\end{equation*}
	
	\begin{equation*}
		\begin{split}
			&\widetilde{sk}_i^{(2)} = \prod_{j=2}^{d} k_{i,j}^{(2)^{y_j}}  \\
			&= (h_1^{-\frac{\theta_{i,2}}{\theta_{i,1}}r} \cdot h_2)^{t_i \cdot y_2} \cdot \cdot \cdot (h_1^{-\frac{\theta_{i,d}}{\theta_{i,1}}r} \cdot h_d)^{t_i \cdot y_d} \\                    
			&= (h_1^{-\frac{\theta_{i,2}}{\theta_{i,1}} \cdot ry_2} \cdot \cdot \cdot h_1^{-\frac{\theta_{i,d}}{\theta_{i,1}} \cdot ry_d})^{t_i} \cdot (h_2^{y_2} \cdot \cdot \cdot h_d^{y_d})^{t_i} \\
			&= (h_1^{-\frac{<\overrightarrow{\theta_i},\overrightarrow{Y}>}{\theta_{i,1}}r} \cdot h_1^{ry_1} \cdot h_2^{y_2} \cdot \cdot \cdot h_d^{y_d})^{t_i}; 
		\end{split}
	\end{equation*}	
	
	{\fontsize{8.0pt}{11.4pt}\selectfont
		\begin{equation*}
			\begin{split}
				& e(sk_{i,1}^{(2)},ct_{1,w}) \cdot \left(\frac{e(\widetilde{sk}_i^{(2)},ct_{1,w})}        {e(ct_{3,w},sk_{i,2}^{(2)})}\right)^{\frac{\theta_{i,1}}{<\overrightarrow{\theta_i},\overrightarrow{Y}>}}  \\
				&= e(g^{\lambda_i} \cdot h_1^{rt_i},\hat{g}^{s_w}) \cdot 
				\left(\frac{e((h_1^{-\frac{<\overrightarrow{\theta_i},\overrightarrow{Y}>}{\theta_{i,1}}r} \cdot h_1^{ry_1} \cdot h_2^{y_2} \cdot \cdot \cdot h_d^{y_d})^{t_i},\hat{g}^{s_w})}{e((h_1^{ry_1}\prod_{i=2}^{d}h_i^{y_i})^{s_w},\hat{g}^{t_i})}\right)^{\frac{\theta_{i,1}}{<\overrightarrow{\theta_i},\overrightarrow{Y}>}} \\
				&= e(g^{\lambda_i} \cdot h_1^{rt_i},\hat{g}^{s_w}) \cdot 
				\left(\frac{e((h_1^{-\frac{<\overrightarrow{\theta_i},\overrightarrow{Y}>}{\theta_{i,1}}r} \cdot h_1^{ry_1} \cdot h_2^{y_2} \cdot \cdot \cdot h_d^{y_d})^{t_i},\hat{g}^{s_w})}{e((h_1^{ry_1} h_2^{y_2} \cdot \cdot \cdot h_d^{y_d})^{s_w},\hat{g}^{t_i})}\right)^{\frac{\theta_{i,1}}{<\overrightarrow{\theta_i},\overrightarrow{Y}>}} \\
				&= e(g^{\lambda_i} \cdot h_1^{rt_i},\hat{g}^{s_w}) \cdot
				{e(h_1,\hat{g})^{- r t_i s_w}} \\
				&= e(g,\hat{g})^{\lambda_i s_w};
			\end{split}
		\end{equation*}
	}
	
	\begin{equation*}
		\begin{split}
			&C_{w,\eta} = \frac{ct_{w,\eta}^{(0)}}{\prod_{i\in I}e(g,\hat{g})^{\pi_i \lambda_i s_w}} \\
			&= \frac{M_{w,\eta} \cdot e(g,\hat{g})^{^{\tilde{\alpha}s_w}}\cdot e(H(M_{w,\eta}\cdot Tag),\hat{g}^r)^{b_w}}
			{e(g,\hat{g})^{\sum_{i\in I} \pi_i \lambda_i s_w}} \\
			&= M_{w,\eta} \cdot e(H(M_{w,\eta}\cdot Tag),\hat{g})^{r b_w};
		\end{split}
	\end{equation*}
	
	\begin{equation*}
		\begin{split}
			&e\left(ct_{w,\eta}^{(1)},sk_{f,2}\right) \\
			&= e\left(H\left(M_{w,\eta}\cdot Tag\right)^{a_w},\hat{g}^{a_v \dot{r}}\right)\\
			&= e\left(H\left(M_{w,\eta}\cdot Tag\right),\hat{g}\right)^{a_w a_v \dot{r}} \\
			&= e\left(H\left(M_{v,\eta}\cdot Tag\right)^{a_v},\hat{g}^{a_w \dot{r}}\right) \\
			&= e\left(ct_{v,\eta}^{(1)},sk_{f,1}\right)
		\end{split}
	\end{equation*}
	
	and	
	
	\begin{equation*}
		\begin{split}
			&\frac{C_{w,\eta}}{e(ct_{w,\eta}^{(1)}\cdot ct_{v,\eta}^{(1)},sk_{f,3})} \\
			&=\frac{M_{w,\eta} \cdot e(H(M_{w,\eta}\cdot Tag),\hat{g})^{r b_w}}{e(H(M_{w,\eta}\cdot Tag)^{a_w} \cdot H(M_{v,\eta}\cdot Tag)^{a_v},\hat{g}^{\frac{r \cdot b_w}{a_w+a_v}})} \\
			&=\frac{M_{w,\eta} \cdot e(H(M_{w,\eta}\cdot Tag),\hat{g})^{r b_w}}{e(H(M_{w,\eta}\cdot Tag)^{a_w+a_v},\hat{g}^{\frac{r \cdot b_w}{a_w+a_v}})} \\
			&= M_{w,\eta}.        
		\end{split}
	\end{equation*}
	
		\begin{figure*}
		\centering
		{\fontfamily{ptm}\selectfont
			\fbox{
				\centering
				\begin{minipage}{\textwidth}
					$Setup(1^\lambda,d,N) \rightarrow (sk,csk_1,csk_2,...,csk_N,pp)$.
					Set $d,N\in\mathbb{N}$ be the size of attribute set of every ciphertext and the number of clients respectively.
					Let $\mathcal{BG}(1^\lambda) \rightarrow (G,\hat{G},G_T,e,p)$.  
					$g,\hat{g}$ denote generators in $G$ and $\hat{G}$ respectively.
					It picks randomly $\overrightarrow{\alpha}=(\alpha_1,...,\alpha_d)^\top \in Z_p^d$, $\overrightarrow{\beta}=(\beta_0,...,\beta_d)^\top \in Z_p^{d+1}$
					and computes $h_i=g^{\alpha_i} \in G$ and $u_j=g^{\beta_j} \in G$ for each $i\in[1,d]$, $j\in[0,d]$.
					Then, two vectors $\overrightarrow{H}=(h_1,...,h_d)^\top$ and $\overrightarrow{U}=(u_0,..,u_d)^\top$ are defined.
					The algorithm picks a random value $\tilde{\alpha} \in Z_p^*$ as the master secret key and calculates $e(g,\hat{g})^{\tilde{\alpha}}$.
					It chooses a hash function $\mathcal{H}: G_T \rightarrow G$.
					In addition, it chooses $a_1,...,a_N,b_1,...,b_N \in Z_p^*$ randomly and defines client encryption key as $csk_k=(a_k,b_k)$,             where $k \in [1,N]$.
					Each client encryption key $csk_k$ is sent to $k^{th}$ corresponding client $CL_k$.
					Hence, the master secret keys and public parameters are respectively
					$$
						msk=(\tilde{\alpha},a_1,...,a_N,b_1,...,b_N) \text{ and }             	pp=(G,\hat{G},G_T,e,p,g,\hat{g},e(g,\hat{g})^{\tilde{\alpha}},\overrightarrow{H},\overrightarrow{U},\mathcal{H},N).
					$$                       
				\end{minipage}
				
			}
			\caption{The Setup algorithm of our MCFE-SI-NAS scheme.}
			\label{FIG:Setup}
		}
	\end{figure*}
	
	\begin{figure*}
		{\fontfamily{ptm}\selectfont
			\fbox{			
				\begin{minipage}{\textwidth}
					$KeyGen(msk,pp,f,\widetilde{\mathbb{A}}) \rightarrow SK_{\widetilde{\mathbb{A}},f}$.
					Given an index function $f=(w,v)$ such that $w, v \in [1,N]$ and $w < v$, it picks randomly $\dot{r}, r\in Z_p$ and calculates $\{sk_{f,1}, sk_{f,2}, sk_{f,3}\}$ as follows using the corresponding client encryption keys $csk_w=\{a_w,b_w\},csk_v=\{a_v,b_v\}$:
					$$sk_{f,1} = \hat{g}^{a_w \cdot \dot{r}},sk_{f,2} = \hat{g}^{a_v \cdot \dot{r}},sk_{f,3} = \hat{g}^{\frac{r \cdot b_w}{a_w + a_v}}.$$ Let $\hat{g}^r,u_0^r,h_1^r$ be public.            
					$\widetilde{\mathbb{A}}$ denotes a non-monotonic access structure such that $\widetilde{\mathbb{A}}=NM(\mathbb{A})$ 
					for the monotonic access structure $\mathbb{A}$,
					where $\mathbb{A}$ is related with a linear secret sharing scheme $\prod$ over an attribute set $P$. 
					By applying $\prod$, it outputs the shares $\{\lambda_i =\mathcal{L}_i\overrightarrow{\zeta}\}_{i \in \{1,2,...,o\}}$ of the master secret key $\tilde{\alpha}$.
					The corresponding party to the share $\lambda_i$ is set as $\breve{x_i} \in P$, where $\breve{x_i}$ is an attribute and can be unprimed(non negated) or primed(negated).
					For each $i \in \{1,2,...,o\}$, the algorithm picks $t_i \in Z_p$ randomly and defines a vector $\overrightarrow{\theta_i} = (\theta_{i,1},...,\theta_{i,d})^\top = (1,x_i,x_i^2,...,x_i^{d-1})^\top$, i.e., $\theta_{i,j}=x_i^{j-1}$.
					The algorithm creates the policy keys $sk_{\widetilde{\mathbb{A}},i}$ as follows. 
					For clarity, $\hat{x}$ denotes a non-negated attribute and $\bar{x}$ stands for a negated attribute.
					$$
						sk_{\widetilde{\mathbb{A}},i} 
						= \left\{ 
						\begin{array}{lcl}
							sk_{i,1}^{(1)} = g^{\lambda_i} \cdot u_0^{rt_i}, sk_{i,2}^{(1)} = \hat{g}^{t_i}, k_{\overrightarrow{\theta_i},i}^{(1)} = g^{t_i \cdot \Delta_{\overrightarrow{\theta_i}}^\top \cdot \overrightarrow{\beta}}  & \mbox{for} & \breve{x_i} = \widehat{x_i} \mbox{ (non-negated)} \\
							sk_{i,1}^{(2)}= g^{\lambda_i} \cdot h_1^{rt_i}, sk_{i,2}^{(2)}= \hat{g}^{t_i}, k_{\overrightarrow{\theta_i},i}^{(2)}=g^{t_i r \cdot \Delta_{\overrightarrow{\theta_i}}^\top \cdot \overrightarrow{\alpha}}  & \mbox{for} &  \breve{x_i} = \overline{x_i} \mbox{ (negated)}
							
						\end{array}
						\right\}, 
					$$        
					where            
					$k_{\overrightarrow{\theta_i},i}^{(1)} = (k_{i,2}^{(1)},...,k_{i,d}^{(1)})
					=((u_1^{-\frac{\theta_{i,2}}{\theta_{i,1}}} \cdot u_2)^{t_i},...,(u_1^{-\frac{\theta_{i,d}}{\theta_{i,1}}} \cdot u_d)^{t_i})
					=g^{t_i \cdot \Delta_{\overrightarrow{\theta_i}}^\top \cdot \overrightarrow{\beta}'},
					$
					$k_{\overrightarrow{\theta_i},i}^{(2)}= (k_{i,2}^{(2)},...,k_{i,d}^{(2)})
					= ((h_1^{-\frac{\theta_{i,2}}{\theta_{i,1}}r} \cdot h_2)^{t_i},...,(h_1^{-\frac{\theta_{i,d}}{\theta_{i,1}}r} \cdot h_d)^{t_i})
					= g^{t_i r \cdot \Delta_{\overrightarrow{\theta_i}}^\top \cdot \overrightarrow{\alpha}},$
					$\overrightarrow{\beta}' = (\beta_1,...,\beta_d)^{\top}$ and $\Delta_{\overrightarrow{\theta_i}} = \left( \begin{array}{cccc}
						-\frac{\theta_{i,2}}{\theta_{i,1}} & -\frac{\theta_{i,3}}{\theta_{i,1}} & ... &-\frac{\theta_{i,d}}{\theta_{i,1}} \\ \multicolumn{4}{c}{I_{d-1}} 
					\end{array}\right)$.
					Hence, decryption keys are $$SK_{\widetilde{\mathbb{A}},f} = (sk_{f,1}, sk_{f,2}, sk_{f,3},\{sk_{\widetilde{\mathbb{A}},i}\}_{\{x_i \in P\}}).$$
				\end{minipage}		
			}
			\caption{The KeyGen algorithm of our MCFE-SI-NAS scheme.}
			\label{FIG:KeyGen}
		}
	\end{figure*}
	
	\begin{figure*}
		{\fontfamily{ptm}\selectfont
			\fbox{			
				\begin{minipage}{\textwidth}
					$Enc(pp,S,Tag,M_k,csk_k) \rightarrow CT_k$.
					Let plaintext set $M_k=\{{M_{k,1},M_{k,2},...,M_{k,l}}\} \in G_T$ is held by corresponding $k^{th}$ client $CL_k$, where every client encrypt same size of plaintext set and $|M_k|=l$. 
					We assume that each client $CL_k$ encrypts its plaintext set $M_k$ under the same attribute set $S$ satisfying $|S|=q<d$.            
					$CL_k$ first defines a polynomial $P_S[X] = \sum_{i=1}^{q+1}(y_i \cdot X^{i-1}) = \prod_{j \in S}(X-j)$ 
					whose coefficients make up the first $q+1$ coordinates of vector $\overrightarrow{Y} = (y_1,..,y_d)^\top$.
					If $q+1<d$, set $y_j = 0$ for $q+2 \leq j \leq d$.
					Then,  $CL_k$ selects a random value $s_k \in Z_p$ 
					and computes $$ct_{1,k} = \hat{g}^{s_k}, ct_{2,k} = (u_0^r \cdot \prod_{i=1}^{d}u_i^{y_i})^{s_k}, ct_{3,k} = (h_1^{ry_1} \prod_{i=2}^{d}h_i^{y_i})^{s_k}.$$
					
					For each $\eta \in [1,2,...,l]$, $CL_k$ uses a label $Tag \in G_T$ to calculate 			$$ct_{k,\eta}^{(0)} = M_{k,\eta} \cdot e(g,\hat{g})^{^{\tilde{\alpha}s_k}}\cdot e(H(M_{k,\eta} \cdot Tag),\hat{g}^r)^{b_k}, ct_{k,\eta}^{(1)} = H(M_{k,\eta}\cdot Tag)^{a_k}.$$
					
					The ciphertext underlying the $M_k=\{M_{k,\eta}\}_{\eta \in \{1,...,l\}}$ are 
					$$CT_k = (\{ct_{k,\eta}^{(0)}\}_{\eta \in [1,2,...,l]},\{ct_{k,\eta}^{(1)}\}_{\eta \in [1,2,...,l]},ct_{1,k},ct_{2,k},ct_{3,k}).$$
					All ciphertexts $\{CT_k\}_{k \in [1,2,...,N]}$ are uploaded to CSP.
				\end{minipage}		
			}
			\caption{The Enc algorithm of the MCFE-SI-NAS scheme.}
			\label{FIG:Enc}
		}
	\end{figure*}
	
	\begin{figure*}
		{\fontfamily{ptm}\selectfont
			\fbox{			
				\begin{minipage}{\textwidth}
					$Dec(pp,CT_w,CT_v,SK_{\widetilde{\mathbb{A}},f}) \rightarrow M_{w}\bigcap M_{v} / \perp$.
					Aggregator requests data from CSP and is responded with the ciphertexts $CT_w$ and $CT_v$.
					Assume that the attribute set $S$ in the ciphertext matches successfully the non-monotonic access structure $\widetilde{\mathbb{A}}$ of the aggregator's decryption key, so that decryption is possible.
					Otherwise, the algorithm outputs $\perp$.            
					Recall that $\widetilde{\mathbb{A}} = NM(\mathbb{A})$, where $\mathbb{A}$ is a monotonic access structure related with a linear secret-sharing scheme $\prod$. 
					Set $S' = N(S) \in \mathbb{A}$ and $I = \{i:\breve{x_i} \in S'\}$. 
					It uses $\overline{x_i}$ to denote negated attribute $\breve{x_i} \in S'$ (i.e., $x_i \notin S$) and  $\widehat{x_i}$ to stand for the non-negated attribute $\breve{x_i} \in S'$ (i.e., $x_i \in S$).
					Since $S' \in \mathbb{A}$, there exists a set of coefficients $\{\pi_{i}\}_{i \in I}$ such that $\sum_{i \in I}(\pi_i \lambda_i) = \tilde{\alpha}$.
					Set the polynomial $P_S[X] = \prod_{j \in S}(X-j) = \sum_{i=1}^{q+1}(y_i X^{i-1})$ whose coefficients are contained in the vector $\overrightarrow{Y} = (y_1,...,y_d)^\top$.          
					The aggregator executes the following decryption procedure as follows.            
					\begin{itemize}
						\item For the non-negated attribute $\widehat{x_i}$, compute                  
						$\frac{e(sk_{i,1}^{(1)} \cdot \prod_{j=2}^{d} k_{i,j}^{{(1)}^{y_j}},ct_{1,w})}{e(ct_{2,w},sk_{i,2}^{(1)})}
						= e(g,\hat{g})^{\lambda_i s_w}$.   
					\end{itemize}        
					\begin{itemize}
						\item  For each negated attribute $\overline{x_i}$, set a vector
						$\overrightarrow{\theta_i} = (1,x_i,...,x_i^{d-1})^\top$ and calculates                  
						$
						e(sk_{i,1}^{(2)},ct_{1,w}) \cdot (\frac{e(\prod_{j=2}^{d} k_{i,j}^{(2)^{y_j}},ct_{1,w})}
						{e(ct_{3,w},sk_{i,2}^{(2)})})^{\frac{\theta_{i,1}}{<\overrightarrow{\theta_i},\overrightarrow{Y}>}}
						= e(g,\hat{g})^{\lambda_i s_w}.  
						$                       
					\end{itemize}
					Then, the algorithm computes the intermediate ciphertext 
					$$
							C_{w,\eta} = \frac{ct_{w,\eta}^{(0)}}{\prod_{i\in I}e(g,\hat{g})^{\pi_i \lambda_i s_w}}                      
							= M_{w,\eta} \cdot e(H(M_{w,\eta}\cdot Tag),\hat{g})^{r b_w}.
					$$
					When the item $M_{w,\eta}$ of the ciphertext $CT_w$ and the item $M_{v,\eta}$ of the ciphertext $CT_v$ are the same, we can derive
					the equation $e(ct_{w,\eta}^{(1)},sk_{f,2}) = e(ct_{v,\eta}^{(1)},sk_{f,1})$ for ciphertext indexing.          
					A set $\Upsilon  = \emptyset$ is initialized.
					Aggregator calculates 
					$$\frac{C_{w,\eta}}{e(ct_{w,\eta}^{(1)}\cdot ct_{v,\eta}^{(1)},sk_{f,3})} = M_{w,\eta}.$$   
					We utilize $M_{\eta}$ to denote the above result.
					It adds all items $M_{\eta}$ into $\Upsilon$. Finally, the algorithm outputs the set $\Upsilon$.
				\end{minipage}		
			}
			\caption{The Dec algorithm of the MCFE-SI-NAS scheme.}
			\label{FIG:Dec}
		}
	\end{figure*}

\section{Security Analysis}
	
	In this section, the security of our MCFE-SI-NAS scheme is formally proved.
	\begin{theorem}	
		The proposed MCFE-SI-NAS scheme is P-IND secure in the random oracle model if the assumptions \cite{Lee2023} and the variant of the $q$-DBDHE assumption hold.
	\end{theorem}
	
	\begin{proof}	
		
		We first define the following intersection function $SIF((M_\eta)_{\eta \in [N]},\mathcal{Q})$.
		Given a tuple $(M_k)_{k \in [N]}$ and a index set $\mathcal{Q}=\{(w,v)\}$,
		$SIF((M_k)_{k \in [N]},\mathcal{Q})$ is able to calculate the collected intersection of the $M_w$ and $M_v$ for every $(w,v)\in \mathcal{Q}$.
		\begin{flushleft}
			\begin{tabular}{l}		
				\hline					 	
				\textbf{Function} $SIF((M_k)_{k \in [N]},\mathcal{Q}) \text { where } \mathcal{Q}=\{(w, v)\}$ \\			
				Set $E_k=\emptyset \text { for all } k \in [N] \text {.}$ \\
				For each $(w,v)\in \mathcal{Q} \text { : }$ \\
				\quad Compute the intersection set $SI=M_w \cap M_v.$ \\
				\quad Add $m \text{ into } E_w \text{ and } E_v.$ \\
				Output $(E_k)_{k\in[N]}$. \\
				\hline
			\end{tabular}
		\end{flushleft}
		
		Then, suppose that a PPT adversary $\mathcal{A}$ attacks our MCFE-SI-NAS scheme with advantage $\epsilon(\lambda)$.
		A simulatior $\mathcal{B}$ is built to play the security game with $\mathcal{A}$ to solve the variant of the $q$-DBDHE and the hard problem assumption in \cite{Lee2023}.
		
		\textbf{Init.}
		$\mathcal{A}$ selects a targeted attribute set $S^*$ that is used to define a vector $\overrightarrow{Y}=(y_1,y_2,...,y_d)^\top$. 
		Hence, a polynomial $P_{S^*}[X]= \prod_{j \in S^{*}}(X-j) = \sum_{i=1}^{q+1}(y_i \cdot X^{i-1}) $ is defined
		whose coefficients make up the first $q+1$ coordinates of vector $\overrightarrow{Y}$,
		where $y_j = 0$ for $q+2 \leq j \leq d$.
		In addition, $\mathcal{A}$ selects two challenging massage tuples $\{M_{1,0}^*,...,M_{N,0}^*\}$ and $\{M_{1,1}^*,...,M_{N,1}^*\}$,
		a targeted tag $Tag^*$, a set of function key queries $Q=\{(w,v)\}$.
		According to the above definition of $\nu^*,\rho,\mathcal{Q}$, $\mathcal{B}$ executes the $J(N,\nu^*,\mathcal{Q})$ and obtains the set $J$.
		Then, $\mathcal{{B}}$ flips a bit $\mu \in \{0,1\}$ randomly and 
		obtains the set $(E_1^*,...,E_N^*)$ by executing the $SIF(\{M_{k,\mu}^*\}_{k\in [N]},\mathcal{Q})$.
		
				 
		Challenger $\mathcal{C}$ flips a fair coin $\psi \in \{0,1\}$.
		If $\psi=0$, set $T=e(g,h)^{z_{d+1}}$ and $Z=g^{ab_\nu^*}$.
		Otherwise, $\psi=1$, set $T \in G_T, Z \in G$.
		Then, $\mathcal{C}$ transfers the tuples $(h,\hat{h},g,\hat{g},z_1,...,z_d, z_{d+2},...,z_{2d} \hat{z}_1,...,\hat{z}_d,\hat{z}_{d+2},...,\hat{z}_{2d},T)$, $(g^a,\left\{g^{b_w}\right\}_{w=1}^N,\left\{g^{a b_\nu}\right\}_{\nu \in J},  \{(\hat{g}^{b_w c_{w, v}}, \hat{g}, \hat{g}^{b_v c_{w, v}}) \}_{(w, v) \in Q},Z)$ to $\mathcal{B}$, where $z_i = g^{(\gamma^i)}$ and $\hat{z}_i=\hat{g}^{(\gamma^i)}$.
		$\mathcal{B}$ will output his guess $\psi'$ on $\psi$.
		
		
		\textbf{Setup.}
		The simulation of the public keys that can be clssified as three types.
		
		(1) \textit{Public key for comment element}. 
		It selects a random value $\vartheta \in Z_p$ 
		and computes $e(z_1,\hat{z}_n)^{\vartheta}=e(g,\hat{g})^{\gamma^{n+1}\vartheta}$ from the tuple.
		Hence, the master key $\tilde{\alpha}$ is  implicitly set as $\tilde{\alpha}=\gamma^{n+1}\vartheta$
		
		(2) \textit{Public keys for non-negated attributes}. 
		It chooses $\delta_0\in Z_p$ randomly 
		and calculates $u_0=g^{\delta_0}\cdot g^{-<\overrightarrow{\gamma},\overrightarrow{Y}>}$.
		$\mathcal{B}$ picks a random vector $\overrightarrow{\delta} \in Z_p^d$ 
		and computes $\overrightarrow{U}'=(u_1,...,u_d)^\top=g^{\overrightarrow{\gamma}}\cdot g^{\overrightarrow{\delta}}$.
		Hence, the value $\overrightarrow{\beta}'$ is implicitly set as $\overrightarrow{\beta}' = (\beta_1,...,\beta_d)^\top = \overrightarrow{\gamma}+\overrightarrow{\delta}$.
		
		(3) \textit{Public keys for negated attributes}. 
		Set $S^*=\{x_1,...,x_q\}$ 
		and the corresponding vectors $\overrightarrow{X}_1,...,\overrightarrow{X}_q$ is defined as $X_\iota = (1,x_\iota,...,x_\iota^{d-1})^\top$.
		$\mathcal{B}$ defines a vector as $\overrightarrow{b_\iota}$ such that 
		$\overrightarrow{b_\iota}^\top\cdot M_{X_\iota}=\overrightarrow{b_{\iota}}^\top \cdot \left( \begin{array}{cccc} -x_\iota & -x_\iota^2 & ... &-x_\iota^{d-1}\\ \multicolumn{4}{c}{I_{d-1}} \end{array}\right)=\overrightarrow{0}$.
		Hence, $\mathcal{B}$ obtains a matrix $\mathbf{B}=(\overrightarrow{b_1}|...|\overrightarrow{b_q}|\overrightarrow{0}|...|\overrightarrow{0})$ 
		Then, $\mathcal{B}$ selects randomly $\overrightarrow{\theta}\in Z_p$ and defines $\overrightarrow{H}$ 
		as $\overrightarrow{H}=g^{\mathbf{B}\overrightarrow{\gamma}}g^{\overrightarrow{\theta}}$.
		Hence, the value $\overrightarrow{\alpha}$ is implicitly set as $\mathbf{B}\overrightarrow{\gamma}+\overrightarrow{\theta}$.

		In addition, $\mathcal{B}$ selects randomly $b_1,...,b_N \in Z_p$, and then prepares a hash list $H$-list for $\{M_{k,\eta}\cdot Tag\}$ that is initially empty. For each $k \in [N]$ and $\eta \in [l]$, $H$-list is updated as follows.
		
		\begin{itemize}
			\item If the $\{M_{k,\eta}\cdot Tag\}$ exists in the $H$-list, $\mathcal{B}$ retrieves the 
			corresponding tuple $(M_{k,\eta}\cdot Tag,u,g^u)$ from the $H$-list and sends $g^u$ to $\mathcal{A}$.
			\item Otherwise, $\mathcal{B}$ calculates:
			\begin{itemize}
				\item{If $k\neq \nu^*$ or $\eta \neq \eta^*$, it selects a random value $u_{k,\eta}' \in Z_p$ and  the tuple $(M_{k,\eta,\mu}^*\cdot Tag^*,u_{k,\eta}',g^{u_{k,\eta}'})$ is added into the list $H$-list.}
				\item{Otherwise $(k=\nu^* \wedge \eta=\eta^*)$, the tuple $(M_{\nu^*,\eta^*,\mu}^*\cdot Tag^*,-,g^a)$ is added into the list $H$-list.}
			\end{itemize}	
		\end{itemize}
		The public parameters 
		$pp=\{G,\hat{G},G_T,e,p,g,\hat{g},e(g,\hat{g})^{\gamma^{n+1}\vartheta},\\ \overrightarrow{U},\overrightarrow{H},H-\text{list},N\}$ are sent to $\mathcal{A}$.
		
		\noindent
		\textbf{Phase-1.}
		$\mathcal{A}$ submits policy key queries with non-monotonic access policy $\tilde{\mathbb{A}}$ with the restriction that
		$\tilde{\mathbb{A}}$ does not match the challenged attribute set $S^*$ i.e., $R(S^*,\tilde{\mathbb{A}}) \neq 1$.
		We assume that $\mathbb{A} = NM(\tilde{\mathbb{A}})$ is defined over a party set $P$, related with a LSSS $\prod$.
		Therefore, we obtain that $R(S',\mathbb{A})\neq 1$, where $S'=NM(S^*)$.
		Let $I=\{i:\breve{x}_i \in S' \}$ be the attribute index in $S'$.
		We denote the $\breve{x}_i$ as the attributes in $S'$ while the underlying $x_i$ as the attributes in $S^*$.
		Set $\mathcal{L} \in \mathbb{Z}_P^{o\times d}$ be the share-generating matrix for $\prod$.
		Since $R(S',\mathbb{A})\neq 1$, $\overrightarrow{1}=(1,0,...,0)^\top$  is linearly independent of the rows of $\mathcal{L}_{S'}$ 
		which is the sub-matrix of $\mathcal{L}$ formed by rows corresponding to attributes in $S'$.
		According to the proposition \ref{proposition}, there exists a coefficient vector $\overrightarrow{\pi}\in Z_p^d$ which satisfies $<\overrightarrow{1},\overrightarrow{\pi}>=\pi_1=1$, $\mathcal{L}_{S'}\cdot \overrightarrow{\pi}=\overrightarrow{0}$ and can be efficiently computed .
		Then, we define a vector $\overrightarrow{v}=\overrightarrow{\zeta}+(\tilde{\alpha}-\zeta_1)\overrightarrow{\pi}$ 
		where $\overrightarrow{\zeta}=(\zeta_1,...,\zeta_d)^\top \in Z_p^d$ are randomly chosen.
		(Note that $\overrightarrow{v}_1 = \tilde{\alpha}$ and that $v_2,...,v_d \in Z_p$ are uniformly distributed.)
		We implicitly set the shares $\lambda_i=\mathcal{L}_i\cdot \overrightarrow{v}$.
		Therefore, we have $\lambda_i=\mathcal{L}_i\cdot\overrightarrow{v}=\mathcal{L}_i \cdot \overrightarrow{\zeta}$ is independent on $\tilde{\alpha}$
		for any $\lambda_i$ such that $\breve{x_i}\in S'$.
		$\mathcal{B}$ calculates the policy keys as follows.
		For ease of description, the $\breve{x_i}$ is classified as $\bar{x_i}$ (negated attribute) and $\widehat{x_i}$ (non-negated attribute).
		
		
		(1) For each $\widehat{x}_i=x_i$, there also exist two situations.
		\begin{itemize}
			\item{
				If $\hat{x_i} \in S^*$, $\lambda_i=<\overrightarrow{\mathcal{L}_i},\overrightarrow{v}>$ is independent on $\tilde{\alpha}$
				and is known by $\mathcal{B}$.
				Hence, $\mathcal{B}$ selects $r,t_i \in Z_p$ randomly and outputs the keys
				$\mathcal{B}$ calculates a tuple
				$$
				(sk_{i,1}^{(1)'} = g^{\lambda_i}u_0^{rt_i},sk_{i,2}^{(1)'}=\hat{g}^{t_i},k_2^{(1)'},\ldots,k_d^{(1)'})
				$$
			}
			
			\item{
				If $\hat{x_i}\notin S^*$, $\lambda_i=<\overrightarrow{\mathcal{L}_i},\overrightarrow{v}>$ is denoted by the form $\lambda_i=\omega_1 \tilde{\alpha}+\omega_2$,
				where the two contants $\omega_1,\omega_2$ are known by the $\mathcal{B}$.
				Then, $\mathcal{B}$ set the $d \times (d-1)$ matrix as follows.
				\begin{equation*}
					\begin{split}
						M_{\overrightarrow{\theta_i}} &= \left( \begin{array}{cccc}
							-\frac{\theta_{i,2}}{\theta_{i,1}} & -\frac{\theta_{i,3}}{\theta_{i,1}} & ... &-\frac{\theta_{i,d}}{\theta_{i,1}} \\ \multicolumn{4}{c}{I_{d-1}} 
						\end{array}\right) \\
						& = \left(
						\begin{array}{cccc}
							-x_i & -x_i^2 & ... & -x_i^{d-1} \\
							\multicolumn{4}{c}{I_{d-1}} 
						\end{array}
						\right).
					\end{split}
				\end{equation*}
				
				Since $\hat{x_i}\notin S^*$, $\mathcal{B}$ defines a vector $\overrightarrow{\xi}=(\xi_1,\ldots,\xi_n)^\top=(1,x_i,x_i^2,\ldots,x_i^{n-1})^\top$ 
				with $\overrightarrow{\xi}^{\top} M_{\overrightarrow{\theta}_i}$ and $<-\overrightarrow{Y},\overrightarrow{\xi}> \neq 0$.
				In addition, set $\tilde{t}=t+\vartheta(\zeta_1\gamma^d+\zeta_2\gamma^{d-1}+...+\zeta_d\gamma)/<\overrightarrow{Y},\overrightarrow{\zeta}>$. 
				Therefore, $\mathcal{B}$ calculates the tuple like 
				$$
				(sk_{i,1},sk_{i,2},k_2,\ldots,k_d) = (g^{\tilde{\alpha}}\cdot u_0^{r\tilde{t}},\hat{g}^{\tilde{t}},g^{\tilde{t_i}M_{\overrightarrow{\theta}_i}\beta'}),
				$$
				where $\beta' = (\beta_1,...,\beta_d)^\top$.
				Then,
				for any vector $\overrightarrow{e} \in Z_p^d$, the $\gamma^{d+1}$ in the $\tilde{t}<\overrightarrow{e},\overrightarrow{\gamma}>$ 
				is $\vartheta <\overrightarrow{e},\overrightarrow{\zeta}>/<\overrightarrow{Y},\overrightarrow{\zeta}>$.
				When $M_{\overrightarrow{\theta_i}}^\top \overrightarrow{\zeta} = \overrightarrow{0}$ and
				$\overrightarrow{e}^\top$ is successfully set as the rows of $M_{\overrightarrow{\theta_i}}^\top$.
				Hence, we have (unknown) $z_{d+1} = g^{(\gamma)^{n+1}}$ can be canceled in $g^{\tilde{t}M_{\theta_i}^{\top}\overrightarrow{\beta}'}$.
				We set $\overrightarrow{Y}=\overrightarrow{e}$ and obtain the following result.
				$$
				g^{\tilde{\alpha}}\cdot u_0^{\tilde{t}} = 
				z_{n+1}^{\vartheta} \cdot (g^{\delta_0} \cdot g^{-<\gamma,\overrightarrow{Y}>})^{\tilde{t}},
				$$
				which can be efficiently calculated, since the coefficient of the $\gamma^{d+1}$ is the $-\vartheta$ in the $-\tilde{t}<\overrightarrow{\gamma},\overrightarrow{Y}>$.
				Given the tuple $(sk_{i,1},sk_{i,2},k_2,\ldots,k_d)$,
				$\mathcal{B}$ can compute the $(sk_{i,1}^{(1)'},sk_{i,2}^{(1)'},k_2^{(1)'},\ldots,k_d^{(1)'})$ using the same way.
			}
		\end{itemize}
		
		(2) For each $\bar{x}_i$, there exist two situations. 
		(Note that according to the above definition, $\bar{x_i} \in S'$ if and only if $x_i \notin S^*$.)
		
		\begin{itemize}
			\item{
				
				If $\bar{x} \notin S'$ ($x_i\in S^*$), the share $\lambda_i = <\overrightarrow{\mathcal{L}}_i,\overrightarrow{v}>$ depends on $\tilde{\alpha}$ 
				and hence can be denoted as $\lambda_i = \omega_1\tilde{\alpha}+\omega_2$, where the two contants $\omega_1,\omega_2 \in Z_p$ are known by $\mathcal{B}$.
				Since $x_i \in S^*=\{S_1,...,S_q\}$, set $x_i=S_\varsigma$ for some $\varsigma \in [1,2,...,q]$.
				Thus, $\mathcal{B}$ chooses $t \in Z_p$ randomly and generate the following tuple:
				\begin{equation*}
					\begin{split}
						&\left(sk_{i,1}^{(2)},sk_{i,2}^{(2)},k_2^{(2)},\ldots,k_d^{(2)}\right) \\
						= 						(g^{\tilde{\alpha}}\cdot & h_{1}^{rt}, \hat{g}^{t},
						(h_{1}^{{-\frac{\theta_{i,2}}{\theta_{i,1}}r}}\cdot h_{2})^{t},\ldots,
						(h_{1}^{{-\frac{\theta_{i,d}}{\theta_{i,1}}r}}\cdot h_{d})^{t}
						),
					\end{split}					
				\end{equation*}					
				where $\overrightarrow{\theta}_i=(\theta_{i,1},...,\theta_{i,d})^\top = \overrightarrow{X}_\varsigma=(1,S_\varsigma,...,S_\varsigma^{(d-1)})=(1,x_i,...,x_i^{n-1})$ and (unknown) $t_i, r \in Z_p$ are selected randomly.
				For every $j \in \{2,...,d\}$, $\mathcal{B}$ picks a random value $t_i' \in Z_p$ and outputs the keys 				
				\begin{equation*}
					\begin{split}
						(sk_{i,1}^{(2)'}= &sk_{i,1}^{(2)^{\omega_1}} \cdot g^{\omega_2} \cdot h_1^{t_i'}, sk_{i,2}^{(2)'}
						= sk_{i,2}^{(2)^{\omega_1}}\cdot \hat{g}^{t_i'}), \\
						& k_{i,j}^{(2)'}=k_{j}^{(2)^{\omega_1}} \cdot (h_1^{-\theta_{i,j}/\theta_{i,1}}\cdot h_j)^{t_i'}.
					\end{split}				
				\end{equation*}			
			}
			
			\item{
				
				If $\bar{x_i} \in S'$ ($x_i \notin S^*$), we have $<\overrightarrow{\mathcal{L}_i},\overrightarrow{\pi}>=0$. 
				Hence, $\overrightarrow{\mathcal{L}_i} \cdot \overrightarrow{\zeta}=\overrightarrow{\mathcal{L}_i}\cdot \overrightarrow{\zeta}$.
				$\mathcal{B}$ selects $r_i \in Z_p$ randomly.
				$\mathcal{B}$ outputs the keys
				$$
				(sk_{i,1}^{(2)'}=g^{\overrightarrow{\mathcal{L}_i} \cdot \overrightarrow{v}}\cdot h_1^{rt_i} ,sk_{i,2}^{(2)'},k_2^{(2)'},\ldots,k_d^{(2)'}).
				$$
				
			}
		\end{itemize}
		
		Then, $\mathcal{B}$ sends the key tuples 
		$(sk_{i,1}^{(1)'},sk_{i,2}^{(1)},\ldots,k_d^{(1)'})$
		and $(sk_{i,1}^{(2)},sk_{i,2}^{(2)},k_2^{(2)},\ldots,k_d^{(2)})$ to $\mathcal{A}$.
		

		
		$\mathcal{A}$ issues a query for the function $f=(w,v)\in \mathcal{Q}$. 
		$\mathcal{B}$ selects a random value $\tilde{b}_w \in Z_p$ and 
		calculates $sk_{f,1}=\hat{g}^{b_wc_{w,v}}$, $sk_{f,2}=\hat{g}^{b_vc_{w,v}}$ and 
		$sk_{f,3}=(\hat{g}^{1/(b_w+b_v)})^{\tilde{b}_w {r}}$ from the given assumption tuple, and sends them to $\mathcal{A}$.
		
		\noindent
		\textbf{Challenge.}
		For every $k\in[N]$ and $\eta \in [l]$, $\mathcal{B}$ generates the ciphertext $ct_{k,\eta}^{(1)}$ as follows.
		
		(1) In the case $k<\rho$, there exist the following three conditions. 
		
		\begin{itemize}
			\item {If $(M_{k,\eta,\mu}^* \in E_k^*) \wedge (M_{k,\eta,\mu}^* = M_{\rho,\eta,\mu})$,
				$\mathcal{B}$ picks the tuple $(M_{k,\eta,\mu}^*\cdot Tag^*,-,g^a)$ from the $H$-list 
				and generates $ct_{k,\eta}^{(1)}=g^{ab_i}$ and $e(H(M_{k,\eta,\mu}^*\cdot Tag^*),\hat{g}^r)^{b_k}=e(g^a,\hat{g}^r)^{b_k}$.}
			\item {If $(M_{k,\eta,\mu}^* \in E_k^*) \wedge (M_{k,\eta,\mu}^* \neq M_{\rho,\eta,\mu})$,
				$\mathcal{B}$ picks the tuple $(M_{k,\eta,\mu}^*\cdot Tag^*,u_{k,\eta}',g^{u_{k,\eta}'})$ from the $H$-list 
				and generates $ct_{k,\eta}^{(1)}=(g^{b_i})^{u_{k,\eta}'}$ and $e(H(M_{k,\eta,\mu}^*\cdot Tag^*),\hat{g}^r)^{b_k}=e(g^{u_{k,\eta}'},\hat{g}^r)^{b_k}$.
			}
			\item {If $M_{k,\eta,\mu}^* \notin E_k^*$,
				the tuple $(M_{k,\eta,\mu}^*\cdot Tag^*,u_{k,\eta}',g^{u_{k,\eta}'})$ is picked from the $H$-list 
				and sets the random value $ct_{k,\eta}^{(1)} \in G$ and $e(H(M_{k,\eta,\mu}^*\cdot Tag^*),\hat{g}^r)^{b_k}=e(g^{u_{k,\eta}'},\hat{g}^r)^{b_k}$.
			}
		\end{itemize}
		
		(2) In the case $k=\rho$, there exist the following four conditions.
		
		\begin{itemize}
			\item{If $(\eta < \delta) \wedge (M_{\rho,\eta,\mu}^* \in E_{\rho}^*)$,
				$\mathcal{B}$ picks the tuple $(M_{\rho,\eta,\mu}^*\cdot Tag^*,u_{\rho,\eta}',g^{u_{\rho,\eta}'})$ from the $H$-list 
				and generates $ct_{\rho,\eta}^{(1)} = (g^{b_\rho})^{u_{\rho,\eta}'}$ and $e(H(M_{\rho,\eta,\mu}^*\cdot Tag^*),\hat{g}^r)^{b_\rho}=e(g^{u_{\rho,\eta}'},\hat{g}^r)^{b_\rho}$ since $M_{\rho,\eta,\mu}^* \neq M_{\rho,\delta,\mu}^*$.
			}
			\item {If $(\eta < \delta) \wedge (M_{\rho,\eta,\mu}^* \notin E_{\rho}^*)$,
				$\mathcal{B}$ picks the tuple $(M_{\rho,\eta,\mu}^*\cdot Tag^*,u_{\rho,\eta}',g^{u_{\rho,\eta}'})$ from the $H$-list 
				and sets the random value $ct_{\rho,\eta}^{(1)} \in G$ and $e(H(M_{\rho,\eta,\mu}^*\cdot Tag^*),\hat{g}^r)^{b_\rho}=e(g^{u_{\rho,\eta}'},\hat{g}^r)^{b_\rho}$.	
			}
			\item {If $\eta = \delta$, 
				$\mathcal{B}$ sets the $ct_{\rho,\eta}^{(1)} = Z$ and $e(H(M_{\rho,\eta,\mu}^*\cdot Tag^*),\hat{g}^r)^{b_\rho}=e(g^{a},\hat{g}^r)^{b_\rho}$ since assuming $M_{\rho,\delta,\mu}^* \notin E_{\rho}^*$.
			}
			\item {If $\eta > \delta$,
				$\mathcal{B}$ picks the tuple $(M_{\rho,\eta,\mu}^*\cdot Tag^*,u_{\rho,\eta}',g^{u_{\rho,\eta}'})$ from the $H$-list 
				and generates $ct_{\rho,\eta}^{(1)} = (g^{b_\rho})^{u_{\rho,\eta}'}$ and $e(H(M_{\rho,\eta,\mu}^*\cdot Tag^*),\hat{g}^r)^{b_\rho}=e(g^{u_{\rho,\eta}'},\hat{g}^r)^{b_\rho}$.
			}
		\end{itemize}
		
		(3) In the case $k>\rho$, there exist the following two conditions.
		
		\begin{itemize}
			\item {If $M_{k,\eta,\mu}^* = M_{\rho,\delta,\mu}^*$
				$\mathcal{B}$ picks the tuple $(M_{k,\eta,\mu}^*\cdot Tag^*,-,g^a)$ from the $H$-list 
				and generates $ct_{k,\eta}^{(1)}=g^{ab_i}$ and $e(H(M_{k,\eta,\mu}^*\cdot Tag^*),\hat{g}^r)^{b_k}=e(g^{a},\hat{g}^r)^{b_k}$.
			}
			\item {If $M_{k,\eta,\mu}^* \neq M_{\rho,\delta,\mu}^*$
				$\mathcal{B}$ picks the tuple $(M_{k,\eta,\mu}^*\cdot Tag^*,u_{k,\eta}',g^{u_{k,\eta}'})$ from the $H$-list 
				and generates $ct_{k,\eta}^{(1)}=(g^{b_i})^{u_{k,\eta}'}$ and $e(H(M_{k,\eta,\mu}^*\cdot Tag^*),\hat{g}^r)^{b_k}=e(g^{u_{k,\eta}'},\hat{g}^r)^{b_k}$.
			}
		\end{itemize}
		
		$\mathcal{B}$ has $u_0 \cdot g^{<\overrightarrow{\beta'},\overrightarrow{Y}>}
		=g^{\vartheta+<\overrightarrow{\delta},\overrightarrow{Y}>} \text{ and }
		g^{<\overrightarrow{\alpha},\overrightarrow{Y}>}=g^{<\overrightarrow{\theta},\overrightarrow{Y}>}$.
		$\mathcal{B}$ flips a coin, and obtains $\mu \in \{0,1\}$.  The challenging ciphertexts of $\{M_{1,\mu}^*,\ldots,M_{N,\mu}^*\}$ are computed as follows.
		\begin{equation*}
			\begin{split}
				&ct_{k,\eta}^{(0)}= M_{k,\eta,\mu}^* \cdot T \cdot e(g^{a},\hat{g}^r)^{\tilde{b_w}}, \\
				&ct_{k,\eta}^{(1)} = Z,
				\quad ct_{1,k}=\hat{h}, 
				\quad ct_{2,k}=h^{\delta_0 r+<\overrightarrow{\delta},\overrightarrow{Y}>},\\&ct_{3,k}=h^{r\overrightarrow{\theta_1} y_1} \cdot h^{y_2 \overrightarrow{\theta_2}} \ldots h^{y_d\cdot \overrightarrow{\theta_d}}
			\end{split}					
		\end{equation*}
		
		If $\psi =0$, then $T=e(g,\hat{h})^{z_{d+1}}$, $Z=g^{ab_{\rho}}$.
		The challenged ciphertext $CT^*$ is 
		\begin{equation*}
			\begin{split}
				&ct_{k,\eta}^{(0)}= M_{k,\eta,\mu}^* \cdot e(g,\hat{g})^{\gamma^{n+1}\vartheta} \cdot e(H(M_{k,\eta,\mu}^*\cdot Tag^*),\hat{g}^r)^{\tilde{b_w}}, \\
				&ct_{k,\eta}^{(1)} = g^{ab_{\rho}},
				\quad ct_{1,k}=\hat{g}^{\vartheta},
				\quad ct_{2,k}=(g^{\delta_0 r}\prod_{i=1}^{d}g^{\delta_i y_i})^{\vartheta}, \\
				&ct_{3,k}=(g^{r\overrightarrow{\theta_1} y_1} \prod_{i=2}^{d} g^{\overrightarrow{\theta_i} y_i})^{\vartheta}.
			\end{split}			
		\end{equation*}
		This is a valid ciphertext for the message $\{M_{1,\mu}^*,\ldots,M_{N,\mu}^*\}$ under attribute sets $S^*$.
		
		Otherwise, if $\psi =1$, $T \in G_T$ and $Z \in G$ are randomly chosen and $CT^*$ hides the message $\{M_{1,0}^*,...,M_{N,0}^*\}$ and $\{M_{1,1}^*,...,M_{N,1}^*\}$.
		
		\noindent
		\textbf{Phase-2.} $\mathcal{B}$ executes repeatedly as it did in \textbf{Phase-1}.
		
		\noindent
		\textbf{Guess.}
		$\mathcal{A}$ submits a guess $\mu'$ on $\mu$.
		$\mathcal{B}$ works as follows.
		
		(1) If $\mu'=\mu$, $\mathcal{B}$ outputs $\psi'=0$.
		
		(2) Otherwise, $\mu'\neq \mu$, $\mathcal{B}$ outputs $\psi'=1$.
		
		If $\psi=0$, $CT^{*}$ is a correct ciphertext, hence $\mathcal{A}$ outputs $\mu' = \mu$ with probability $\frac{1}{2} + \epsilon(\lambda)$.
		When $\mu'=\mu$, $\mathcal{B}$ outputs $\psi'=0$.
		We have $Pr[\psi = \psi'|\psi =0]$=$\frac{1}{2} + \epsilon(\lambda)$. 
		
		If $\psi=1$, $CT^{*}$ is the one time pad of $\{M_{1,0}^*,...,M_{N,0}^*\}$ and $\{M_{1,1}^*,...,M_{N,1}^*\}$, hence $\mathcal{A}$ outputs $\mu' \neq \mu$ with probability $\frac{1}{2}$.
		When $\mu' \neq \mu$, $\mathcal{B}$ outputs $\psi'=1$.
		We have $Pr[\psi = \psi'|\psi =1]$=$\frac{1}{2}$.		
		Hence, the advantage that $\mathcal{B}$ can break the variant of the $q$-DBDHE assumption and the assumption in \cite{Lee2023} is 
		\begin{align*}
			& \left | \frac{1}{2}\times Pr[\psi = \psi'|\psi =0] - \frac{1}{2} \times Pr[\psi = \psi'|\psi =1] \right | \\
			& \geqslant \frac{1}{2} Pr[\psi = \psi'|\psi =0] - \frac{1}{2} Pr[\psi = \psi'|\psi =1] \\
			&= \frac{1}{2} \times (\frac{1}{2} + \epsilon(\lambda)) - \frac{1}{2} \times \frac{1}{2} \\
			&= \frac{\epsilon(\lambda)}{2}
		\end{align*}
		
		\end{proof}
		
\section{Efficiency Analysis}

	The MCFE-SI-NAS scheme is implemented on the Lenovo Y9000K laptop with an Intel i7-11800H CPU and 32M RAM.
	For implementing the bilinear map, we utilize java pairing-based cryptography (JPBC) library \cite{JPBC}
	which is an open source library written in Java and supports many types of elliptic curves and other algebraic curves.
	We select the type F curve $y^2 = x^3 + b$ for supporting the pairing operations which 
	is a pairing-friendly curve and is able to support Type-III pairing.
	We implement the each algorithm of our MCFE-SI-NAS scheme and show the computation costs in the Fig 6.
	
	Set size of attribute set in ciphertexts is $d=10$.
	Let $N$ be the number of the clients and $l$ stand for the size of the plaintext set.
	We consider the following three cases during implementing the presented scheme.
	Case-I. $N=5$, $l=5$;
	Case-II. $N=10$, $l=5$;
	Case-III. $N=10$, $l=10$.
	Each algorithm runs five times, and the average value is taken as experimental result. 
	
	In $Setup$ algorithm, TA is responsible for generating  $csk_k \in Z_p$ for each client, which does not  involves pairing or exponential calculation and is irrelvance to plaintext size. Implementation result shows that $Setup$ algorithm takes about 558.6 ms, 559.6 ms and 560.4 ms in Case-I, Case-II and Case-III, respectively.
	
	The $KeyGen$ algorithm is executed by the TA for calculating the decryption keys for aggregator, and index function in decryption key is denoted by $f=(w,v)$, which is a set of fixed size.
	$KeyGen$ algorithm takes about 348.4 ms, 346.6 ms and 362.6 ms in three cases, respectively.
	
	In $Enc$ algorithm, each client encrypts independently their plaintext sets.
	It takes about 3256.8 ms, 6309.8 ms and 12530.6 ms in Case-I, Case-II and Case-III, respectively.
	The computation costs of the $Enc$ algorithm grows linearly with $l$ and $N$.
	
	After receiving the ciphertext sets of a pair of clients, aggregator executes the $Dec$ algorithm for 
	obtaining the plaintext intersection of this pair of clients.
	The computation costs of the $Dec$ algorithm are about 2090.8 ms in Case-I, 2144.4 ms in   Case-II, and 2802.6 ms in Case-III, which is linear with the size of the plaintext set.
	
	\begin{figure*}
		\centering
		\begin{subfigure}[b]{0.45\textwidth}
			\centering
			\includegraphics[width=\textwidth]{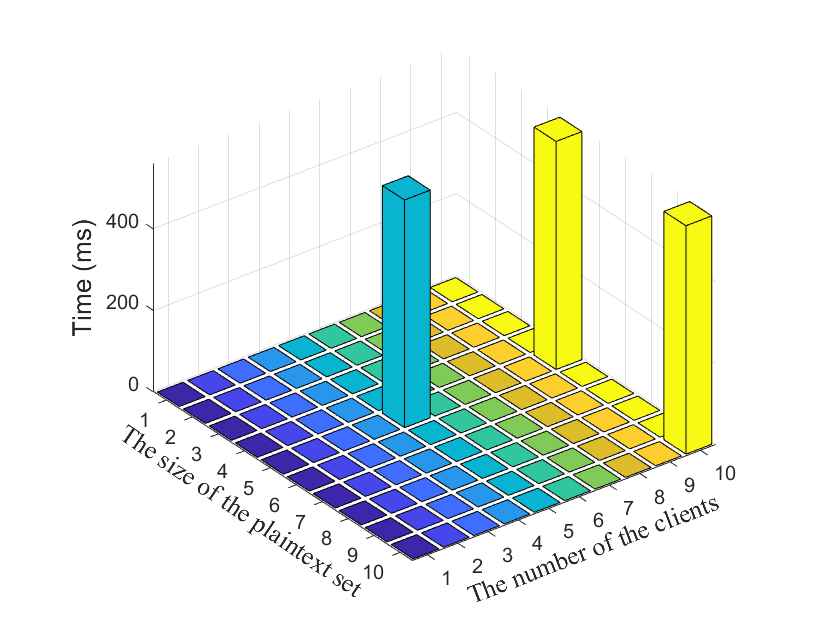}
			\caption{Setup algorithm}
			\label{time_of_setup}
		\end{subfigure}
		\hfill
		\begin{subfigure}[b]{0.45\textwidth}
			\centering
			\includegraphics[width=\textwidth]{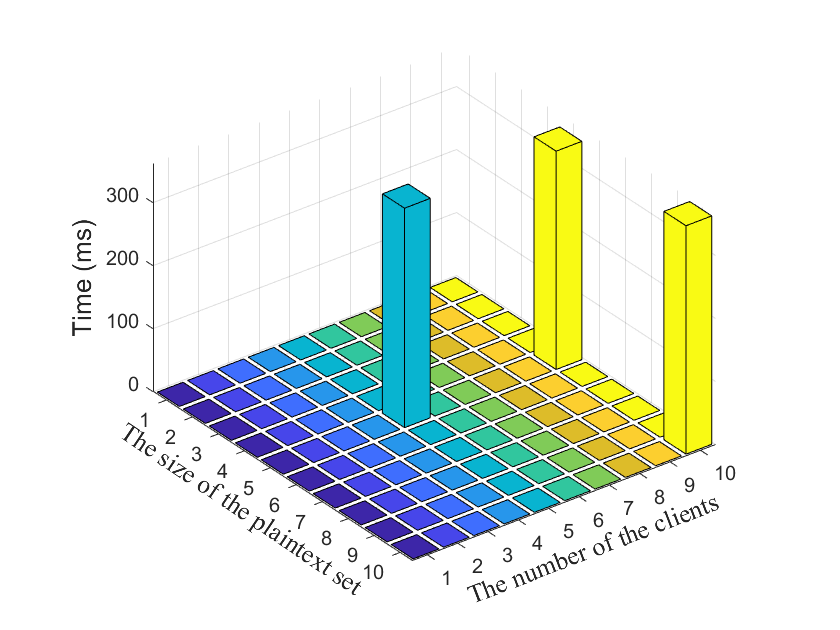}
			\caption{KeyGen algorithm}  
			\label{time_of_keygen}
		\end{subfigure}
		\hfill
		\begin{subfigure}[b]{0.45\textwidth}
			\centering
			\includegraphics[width=\textwidth]{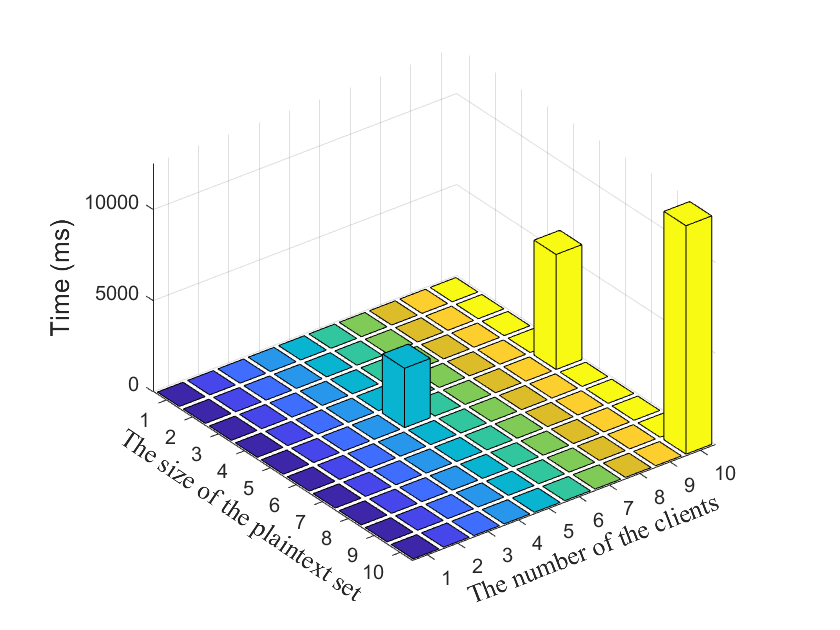}
			\caption{Enc algorithm}  
			\label{time_of_enc}
		\end{subfigure}
		\hfill
		\begin{subfigure}[b]{0.45\textwidth}
			\centering
			\includegraphics[width=\textwidth]{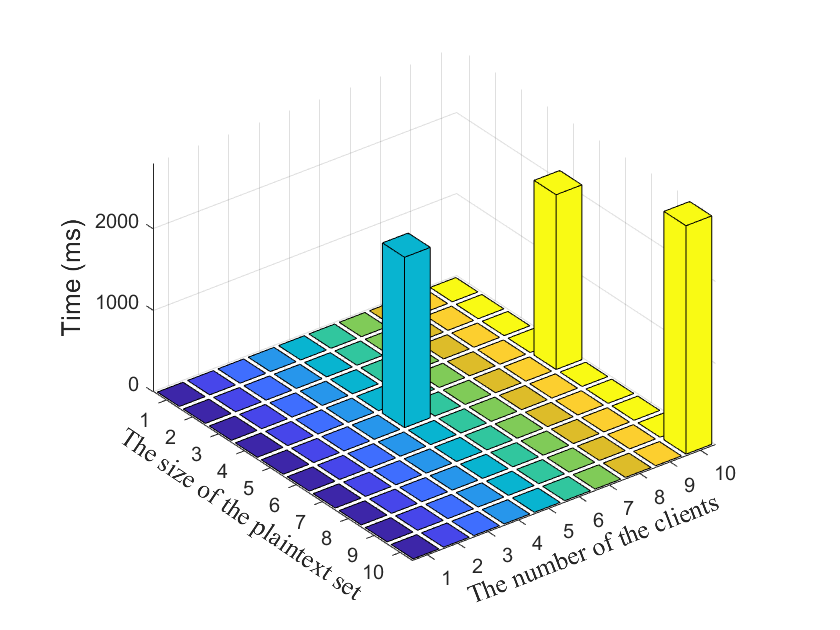}
			\caption{Dec algorithm}
			\label{time_of_dec}
		\end{subfigure}
		
		\caption{The computation cost of our MCFE-SI-NAS scheme}
		\label{fig:time_algorithms}
	\end{figure*}	

\section{Conclusion}

	In this paper, we presented a MCFE-SI-NAS scheme that supports the non-monotonic access structures and set intersection operations.
	The proposed scheme allows each client co-exists and encrypts independently, which is suitable for FL environment. 
	Our MCFE-SI-NAS scheme allows the aggregator to aggregate ciphertexts, and only learn the intersection of private sets held by the specified clients without revealing anything else about plaintexts.
	The designed non-monotonic access structures support any access formula
	including  "AND" gate, "OR" gate, "NOT" gate and threshold policy, which is more flexible than monotonic access policy.
	We first gave the formal definition and security model of the MCFE-SI-NAS scheme and described a concrete construction.
	We proved the security of the proposed scheme in the random oracle model and also provide performance analysis.
	
\printcredits

\section*{Acknowledgement}
	 This work was supported by the National Natural Science Foundation of China (Grant No. 62372103, 61972190), the Natural Science Foundation of Jiangsu Province (Grant No. BK20231149), the Jiangsu Provincial Scientific Research Center of Applied Mathematics (Grant No. BK202330\\02),  and the Start-up Research Fund of Southeast University (Grant No. RF1028623300).

\bibliographystyle{cas-model2-names}

\bibliography{cas-refs}

\end{document}